\documentclass[a4paper,UKenglish]{article}
\usepackage[margin=1.0in]{geometry}
\usepackage{mathtools}
\usepackage{xfrac}
\usepackage{xspace}
\usepackage{xargs}
\usepackage{amssymb}
\usepackage{amsfonts}
\usepackage{enumerate}
\usepackage{amsthm}
\usepackage[textsize=tiny]{todonotes}

\usepackage{color}
\usepackage[T1]{fontenc}
\usepackage[utf8]{inputenc}
\usepackage{cite}
\usepackage{calrsfs} %caligraphic
\usepackage{hyperref}
\usepackage{authblk}
\usepackage{paralist}

\allowdisplaybreaks %for align to break

\newtheorem{theorem}{Theorem}[section]
\newtheorem{lemma}[theorem]{Lemma}

\newtheorem{observation}{Observation}
\newtheorem{definition}{Definition}

\DeclarePairedDelimiter{\floor}{\lfloor}{\rfloor}

\newcommand{\kdots}{, \dots, }

\newcommand{\de}{d}

\newcommand{\dmax}{d_{max}}
\newcommand{\dmin}{d_{min}}
\newcommand{\davg}{d_{avg}}

\newcommand{\hittime}{t_{hit}}
\newcommand{\meettime}{t_{meet}}

\newcommand{\Exp}{\ensuremath{\operatorname{\mathbf{E}}}}

\newcommand{\Var}{\ensuremath{\operatorname{\mathbf{Var}}}}
\newcommand{\vol}{\ensuremath{\operatorname{\mbox{vol}}}}

\newcommand{\logn}{\log n }

\newcommand{\pot}[1]{\Psi(S_{{#1}})}
\newcommand{\potfixed}[1]{\Psi(s_{{#1}})}

\DeclareUnicodeCharacter{00A0}{ }
\newcommand\numberthis{\addtocounter{equation}{1}\tag{\theequation}}

\title{Bounds on the Voter Model in Dynamic Networks}
\author[1]{Petra Berenbrink\thanks{petra@sfu.ca}}
\author[2]{George Giakkoupis\thanks{george.giakkoupis@inria.fr}}
\author[2]{Anne-Marie Kermarrec\thanks{anne-marie.kermarrec@inria.fr}}
\author[1,3]{\\Frederik Mallmann-Trenn\thanks{fmallman@sfu.ca}}
\affil[1]{Simon Fraser University, Burnaby, Canada}
\affil[2]{INRIA, Rennes, France}
\affil[3]{École Normale Supérieure, Paris, France}

\begin{document} %Needs hyper ref and authblk
\maketitle

\begin{abstract}
In the \emph{voter model}, each node of a graph has an opinion, and in every round each node chooses independently a random neighbour and adopts its opinion. We are interested in the \emph{consensus time}, which is the first point in time where all nodes have the same opinion.
We consider dynamic graphs
in which the edges are rewired in every round (by an adversary) giving rise to the graph sequence $G_1, G_2, \dots $, where we assume that $G_i$ has conductance at least $\phi_i$.
We assume that the degrees of nodes don't change over time as one can show that the consensus time can become super-exponential otherwise.
In the case of a sequence of $d$-regular graphs, we obtain asymptotically tight results.
Even for some static graphs, such as the cycle, our results improve the state of the art.
Here we show that the expected number of rounds until all nodes have the same opinion is bounded by $O(m/(\dmin \cdot \phi))$, for any  graph with $m$ edges, conductance $\phi$, and degrees at least $\dmin$.
%This bound is the best possible when either $\phi$ or $d$ are constant (independent of $n$).
In addition, we consider a \emph{biased} dynamic voter model, where each opinion $i$ is associated with  a probability $P_i$, and when a node chooses a neighbour with that opinion, it adopts opinion $i$ with probability $P_i$ (otherwise the node keeps its current opinion).
We show for any regular dynamic graph, that if there is an $\epsilon>0$ difference between the highest and second highest opinion probabilities, and at least $\Omega(\log n)$ nodes have initially the opinion with the highest probability, then all nodes adopt w.h.p.\ that opinion.
We obtain a bound on the convergence time, which becomes $O(\log n/\phi)$ for static graphs.
\end{abstract}

\setcounter{footnote}{0}
\section{Introduction}
In this paper, we investigate the spread of opinions in a connected and undirected graph using the \emph{voter model}. The standard voter model works in synchronous rounds and is defined as follows. At the beginning, every node has one opinion from the set $\{ 0, \dots, n-1\}$, and in every  round, each node chooses one of its neighbours uniformly at random and adopts its opinion.
In this model, one is usually interested in the {\em consensus time} and the {\em fixation probability}. The consensus time  is the number of rounds it takes until all nodes have the same opinion. The fixation probability of opinion $i$ is the probability that this opinion prevails, meaning that all other opinions vanish.  This probability  is known to be proportional to the sum of the degrees of the nodes starting with opinion~$i$~\cite{HP01,NIY00}.

%The voter model has many applications.
%It can be used to model the spread of opinions in a social network. It can also be used to study the  influence of voters during an election  and it finds its application in leader election (See, e.g., \cite{AW04} for a definition).

The voter model is the dual of the {\em coalescing random walk model} which can be described as follows. Initially, there is a pebble on every node of the graph. In every round, every pebble chooses a neighbour uniformly at random and moves to that node. Whenever two or more pebbles meet at the same node, they are merged into a single pebble which continues performing a random walk. The process terminates when only one pebble remains. The time it takes until only one pebble remains is called  {\em coalescing time}. It is known that the coalescing time for a graph $G$ equals the  consensus time of the voter model on $G$ when initially each node has a distinct opinion~\cite{AF14,LN07}.

\medskip

In this paper we consider the voter model and a {\em biased}
variant where the opinions have different popularity.
We express the consensus time as a function of the graph \emph{conductance} $\phi$.

We assume a dynamic graph model
 where the edges of the graph can be rewired by an adversary in every round, as long as the adversary respects the given degree sequence and the given conductance for all generated graphs.
We show that consensus is reached with constant probability after  $\tau$ rounds, where  $\tau$ is the first round such that the sum of conductances up to round  $\tau$ is at least $m /\dmin$, where $m$ is the number of edges.
For static graphs the above  bound simplifies to
$O(m/(\dmin \cdot \phi))$, where $\dmin$ is the minimum degree.

For the biased model we assume a \emph{regular} dynamic  graph $G$.
Similar to \cite{LN07,KT08} the opinions have a {\em popularity}, which is expressed as a probability with which nodes adopt  opinions.
Again, every node chooses one of its neighbours uniformly at random, but this time it adopts the neighbour's opinion with a probability that equals the popularity of this opinion (otherwise the node keeps its current opinion).
We  assume that the popularity of the most popular  opinion is 1, and every other opinion has a popularity of at most $1-\epsilon$ (for an arbitrarily small but constant $\epsilon > 0$).
We also assume that at least $\Omega(\log n)$ nodes start with the most popular opinion.
 Then we show that
the most popular opinion prevails w.h.p.\footnote{An event happens \emph{with high probability (w.h.p.)} if its probability is at least $1-{1}/{n}$.}  after  $\tau$ rounds, where  $\tau$ is the first round such that the sum of conductances up to round
$\tau$ is of order $O(\log n)$.
For static graphs the above  bound simplifies as follows:
the most popular opinion prevails w.h.p.\  in $O(\logn/\phi)$ rounds, if
at least $\Omega(\log n)$ nodes start with that opinion.

\subsection{Related work}
A sequential version of the voter model was introduced in~\cite{HL75} and can be described as follows. In every round, a single node is chosen uniformly at random and this node changes its opinion to that of a random neighbour. The authors of~\cite{HL75}  study infinite grid graphs. This was generalised to arbitrary graphs in~\cite{DW83} where it is shown among other things that the probability for opinion $i$ to prevail is proportional to the sum of the degrees of the nodes having opinion $i$ at the beginning of the process.

The standard voter model was first analysed in~\cite{HP01}.
%The authors show that, similar to the sequential process,  the probability for opinion $i$ to prevail is proportional to the sum of the degrees of the nodes with  opinion $i$ at the beginning of the process.
The authors of~\cite{HP01} bound the expected coalescing time (and thus the expected consensus time) in terms of the  expected meeting time $\meettime$ of two random walks  and show a bound of $O(\meettime \cdot \logn)=O(n^3\log n)$.
Note that the meeting time is an obvious lower bound on the coalescing time, and thus a lower bound on the consensus time when all nodes have distinct opinions initially.
The authors of~\cite{CEOR12} provide an improved upper bound of $O\big(\tfrac{1}{1-\lambda_2} (\log^4n+ \rho)\big)$ on the expected coalescing time for any graph $G$, where $\lambda_2$ is the second eigenvalue of the transition matrix of a random walk on $G$, and $\rho = \big(\sum_{u\in V(G)}d(u)\big)^2/ \sum_{u\in V(G)}d^2(u)$ is the ratio of the square of the sum of node degrees over the sum of the squared degrees.
The value of $\rho$ ranges from $\Theta(1)$, for the star graph, to $\Theta(n)$, for regular graphs.

%Furthermore, the authors show that for any connected graph the meeting time is bounded by $O(n^3)$.

The  authors of~\cite{L85,AF14,LN07} investigate coalescing random walks in a continuous setting where the movement of the pebbles are modelled by independent Poisson processes with a rate of 1. In~\cite{AF14}, it is shown a lower bound of $\Omega(m/\dmax)$ and an upper bound of $O( \hittime \cdot \log n)$ for the expected coalescing time. Here $m$ is the number of edges in the graph, $\dmax$ is the maximum degree, and $\hittime$ is the (expected) hitting time.
%\footnote{The hitting time of a pair of nodes $u$ and $v$ is the expected time until a random walk starting at $u$ hits $v$. The hitting time of a graph is the maximum hitting time over all pairs of nodes in the graph.}
In~\cite{O12}, it is shown that the expected coalescing time is bounded by $O(\hittime)$.

In\cite{LN07} the authors consider the biased voter model in the continuous setting and two opinions.
They show that for  $d$-dimensional lattices the probability for  the less popular  opinion to prevail is exponentially small.
In~\cite{KT08}, it is shown that in this setting the expected consensus time is exponential for the line.

The authors of~\cite{CER14} consider a modification of the standard voter model with two opinions,
which they call {\em two-sample voting}.
In every round, each node chooses two of its neighbours randomly and  adopts
their opinion only if they both agree. For regular graphs
and random regular graphs, it is shown that two-sample voting has a consensus time of $O(\log n)$ if the initial imbalance between the nodes having the two opinions is large enough.
There are several other works on the setting where every node contacts in every round two or more neighbours  before adapting its opinion  \cite{AD15,CG14,CERRS15,EFKMT16}.
%For regulates graphs the initial imbalance has to fulfil \gtodo{0,1 as notation for opinions}$(|OP_1|-|OP_2|/n\ge K \lambda_2$, where $k$ is a constant, $|OP_1|\ge |OP_2|$, and $\lambda_2$ the second largest eigenvalue of the transition matrix.

There are several other models which are related to the voter model, most notably the {\em Moran process} and
{\em rumor spreading} in the phone call model.
In the case of the Moran process, a population resides on the vertices of a graph. The initial population consists of one
mutant with fitness $r$ and the rest of the nodes are non-mutants with fitness 1.
In every round, a node is chosen at random with probability proportional to its fitness. This node then reproduces by placing a copy of itself on a randomly chosen neighbour, replacing the individual that was there.
The main quantities of interest are the probability that the mutant occupies the whole graph (fixation) or vanishes (extinction), together with the time before either of the two states is reached (absorption time).
There are several publications considering the fixation probabilities~\cite{HV11,MS13,DGMRSS14}.
%For any graph it is known that the expected absorption time for an advantageous mutation ($r>1$) is $O(n^4)$~\cite{DGMRSS14}.
%The authors of~\cite{DGRS13}  show that for regular directed graphs the expected absorption time is $\Omega(n \log n)$ and $O(n^2)$, and give improved bounds for other families of graphs based on the expansion of these graphs.

Rumor spreading in the phone call model
works as follows. Every node $v$ opens a channel to a randomly chosen neighbour $u$. The channel can be used for transmissions in both directions. A transmission from $v$ to $u$ is called {\em push} transmission and a transmission from $u$ to $v$ is called {\em pull}. There is a vast amount of papers analysing rumor spreading on different graphs.
%The author of
%\cite{G14} shows that $O(\log^2 n/\alpha)$ rounds suffice w.h.p. for broadcasting in a network with expansion $\alpha$.
The result that is most relevant to ours is that broadcasting of a message in the whole network is completed in $O(\log n/\phi)$ rounds w.h.p, where $\phi$ is the conductance (see Section~\ref{sec:modelVoter} for a definition) of the network.
In \cite{GSS14}, the authors study rumor spreading in dynamic networks, where the edges in every round are distributed by an adaptive adversary. They show that broadcasting terminates w.h.p.\ in a round $t$ if the sum of conductances up to round $t$ is of order $\log n$.
%
%Note that the voter model uses {\sc pull} transmissions whereas the Moran process uses
%{\sc push} transmissions. Hence, broadcasting in the {\sc pull} model is similar to biased voting with two opinions having popularity zero and one. In the beginning, only one node has the opinion with popularity one, all other nodes have the other opinion. Please also note that analyzing the voter process is much more challenging as analyzing broadcasting in the phone call model. The process is much more complex since nodes can change their opinion several times during the process.
%
%rounds. This implies, that if there are only two opinions, then consensus is reached w.h.p. after $O(\logn/\phi)$ time rounds.
%We also show that if initially there are logarithmic number of nodes sharing the  most popular opinion,  then  this opinion prevails w.h.p.
%Note that this model can be regarded as an extension  of broadcasting or rumor spreading using {\sc pull}-transmissions \cite{DGHI87}. Here, only one node is informed at the beginning and in every round
%all informed nodes choose a neighbor uniformly at random and they send the rumor to that node. Hence, broadcasting in the {\sc pull} model is similar to biased voting with two opinions having popularity zero and one. At the beginning, only one node has the  opinion with popularity one, all other nodes have the other opinion.
%Note that  the required time to spread the rumor in using {\sc pull}-transmissions is
%$\Theta(\log n)$, which means that the analysis of the consensus time is tight for this model.
%
%
Here, the sequence of graphs $G_1, G_2, \dots$ have the same
vertex set of size $n$, but possibly distinct edge sets. The authors assume that the degrees and the conductance may change over time. 
We refer the reader to the next section for a discussion of the differences.
Dynamic graphs have received ample attention in various areas \cite{AKL08,KLO10,PSSS13,LLMSW12}.
%The most related dynamic problem studied is probably \cite{GSS14}, where the authors consider rumor spreading in dynamic graphs.
%
%a dynamic setting, given by a sequence of graphs G1, G2, . . . with the same
%vertex set of size n, but possibly distinct edge sets.

\subsection{Model and New Results} \label{sec:modelVoter}
In this paper we show results for the standard voter model  and biased voter model
in dynamic graphs.
Our protocols work in synchronous steps. The consensus time $T$ is defined at the first time step at which all nodes have the same opinion.

\paragraph*{Standard Voter Model.}
Our first result concerns the standard voter model in dynamic graphs.  Our protocol works as follows.
 In every synchronous time step every node chooses a neighbour u.a.r. and adopts its opinion with probability $1/2$.\footnote{The factor of $1/2$ ensures that the process converges on bipartite graphs.}

We assume that the dynamic graphs $\mathcal{G}=G_1, G_2,\ldots$ are generated by an adversary. We assume that each graph has $n$ nodes and the nodes are numbered from $1$ to $n$.  The sequence of conductances
$\phi_1,\phi_2,\ldots$ is given in advance, as well as a degree sequence
$d_1,d_2,\ldots, d_n$.
The adversary is now allowed to create every graph $G_i$  by redistributing the edges of the graph.   The constraints are that each  graph $G_i$  has to have conductance $\phi_i$ and node $j$ has to have degree $d_j$ (the degrees of the nodes do not change over time).
Note that the sequence of the conductances is fixed and, hence,  cannot be regarded as a random variable in the following.
For the redistribution of the edges we assume that the adversary
 knows the distribution of all opinions during all previous rounds.

Note that our model for dynamic graphs is motivated by the model presented in  \cite{GSS14}. They allow the adversary to determine the edge set at every round, without having to respect the node degrees and conductances.

We show (Observation~\ref{dontlettheadversarychangethedegrees}) that, allowing the adversary to change the node degrees over time can results in super-exponential voting time. Since this changes the behaviour significantly, we assume that the degrees of nodes are fixed.
Furthermore, in contrary to \cite{GSS14}, we assume that (bounds on) the conductance of (the graph at any time step) are fixed/given beforehand.
Whether one can obtain the same results, if the conductance of the graph is determined by an adaptive adversary remains an open question.
The reason we consider an adversarial dynamic graph model is in order to understand how the voting time can be influenced in the worst-case.
Another interesting model would be to assume that in every round the nodes are connected to random neighbours. One obstacle  to such a model seems to be to guarantee  that neighbours are chosen u.a.r.\ and the degrees of nodes do not change.
For the case of regular random dynamic graphs our techniques easily carry over since the graph will have constant conductance w.h.p.\ in any such round since the graph is essentially a random regular graph in every round.

For the (adversarial) dynamic model we show the following result bounding the consensus time $T$.

\begin{theorem}[upper bound]  \label{thm:voter}
Consider the Standard Voter model and in the dynamic graph %generation 
model. Assume
$\kappa \leq n$  opinions are  arbitrarily distributed over the nodes of $G_1$.
Let $\phi_t$ be a lower bound on the conductance at time step $t$.
Let $b>0$ be a suitable chosen constant.
Then, with a probability of $1/2$ we
have that $T\le \min\{\tau,\tau'\}$, where
%Assume  $\tau = \infty$ and $\tau' = \infty$ if no such round exists.

\begin{enumerate}[(i)]
\item $\tau$ is the first round so that
$\sum_{t=1}^\tau \phi_t \geq b \cdot m/\dmin$. \emph{(part 1)}
\item  $\tau'$ is the first round so that $\sum_{t=1}^{\tau'} \phi_t^2 \geq b \cdot n \log n.$  \phantom{2}\emph{(part 2)}

 \end{enumerate}
For static graphs ($G_{i+1}=G_i$ for all $i$), we have $T\leq \min \{m/(\dmin \cdot \phi), n\log n/\phi^2\}$.
\end{theorem}
%
%What we should say.
%\begin{itemize}
%\item We are the first to consider dynamic graphs in the voter model. Coalescing random walk approaches won't  won't work here? (depends on adversary)
%\item we use a potential argument.
%\item For static graphs holds for coalescing random walks
%\item For regular graphs we are better than the existing literature.
%\item We are tight if d or $\phi$ is a constant
%\item We are tight for dynamic graphs where the adversary can rearrange edges.
%\end{itemize}
%

For \emph{static} $d$-regular graphs, where the graph doesn't change over time, the above bound becomes $O(n/\phi)$, which is tight when either $\phi$ or $d$ are constants (see Observation \ref{lem:lowerbound}).
Theorem \ref{thm:voter} gives the first tight bounds for cycles and circulant graphs $C_n^{k}$ (node $i$ is adjacent to the nodes $i \pm 1, \ldots, i \pm k \mod n$) with degree $2k$ ($k$ constant). For these graphs the consensus time is $\Theta(n^2)$, which matches our upper bound from Theorem \ref{thm:voter}.\footnote{The lower bound of $\Omega(n^2)$ follows from the fact that two coalescing random walks starting on opposite sites of a cycle require in expectation time $\Omega(n^2)$ to meet.}
For a comparison with the results of~\cite{CEOR12} note that $\phi^2 \leq 1-\lambda_2 \leq 2\phi$.
In particular, for the cycle $\phi = 1/n$ and $1/(1-\lambda_2) = \Theta(1/n^2)$.
Hence, for this graph, our bound is by a factor of $n$ smaller. %\gtodo{compare also with hitting time? see Aldous' book Proposition 4.42}
Note that, due  to the duality between the voter model and coalescing random walks, the  result also holds for the coalescing time.
In contrast to~\cite{CEOR12,CER14}, the above result  is shown using a potential function argument,  whereas the authors of \cite{CEOR12,CER14} show their results for  coalescing random walks and fixed graphs. The advantage of analysing the process directly is, that our techniques allow us to obtain the results for the dynamic setting.

\medskip

The next result shows that the bound of Theorem \ref{thm:voter} is asymptotically tight if the adversary  is allowed to change  the node degrees over time.

\begin{theorem}[lower bound]\label{st:adaptivelowerbound}
Consider the Standard Voter model in the dynamic graph %generation 
model.
%Assume that the adversary is allowed to change the node degrees of all nodes in every round.
Assume that
$\kappa \leq n$  opinions are  arbitrarily distributed over the nodes of $G_1$.
Let $\phi_t$ be an upper bound on the conductance at time step $t$.
Let $b>0$ be a suitable constant and
assume $\tau''$ is the first round such  that
$\sum_{t=1}^{\tau''} \phi_t \geq b n.
$
Then,  with a probability of at least $1/2$,  there are still nodes with different opinions in $G_{\tau''}$.

\end{theorem}

\paragraph*{Biased Voter Model}
In the \emph{biased} voter model  we again assume that there are
$\kappa\le n$ distinct opinions initially. For $0\le i\le \kappa-1$,
opinion $i$ has popularity $\alpha_i$ and we assume that $\alpha_0=1 > \alpha_1\ge \alpha_2\ge \ldots \ge \alpha_{\kappa-1}$. We call opinion $0$ the \emph{preferred opinion}.
The process works as follows.
In every round, every node chooses a neighbour uniformly at random and adopts its opinion $i$ with probability $\alpha_i$.

We assume that the dynamic $d$-regular graphs $\mathcal{G}=G_1, G_2,\ldots$ are generated by an adversary. We assume that the sequence of $\phi_t$ is given in advance, where $\phi_i$ is a lower bound on the conductance of
$G_i$. The adversary is now allowed to create the sequence of graphs by redistributing the edges of the graph in every step.  The constraints are that each  graph $G_i$ has $n$ nodes  and has to have conductance at least $\phi_i$. Note that we assume that the sequence of the conductances is fixed and, hence,  it is not a random variable in the following.

The following result shows  that consensus is reached considerably faster in the biased voter model, as long as the bias $1-\alpha_1$ is bounded away from 0, and at least a logarithmic number of nodes have the preferred opinion initially.

%\begin{theorem}\label{thm:ub-bvoter}
%Let $G = (V, E)$ be any connected regular graph with $n$ nodes, $k\le n$ opinions, and conductance $\phi$.
%Let $\beta=\beta(\alpha_1)$ be a constant.
%Consider a biased voter model on $G$ with $\alpha_1\le(1-\epsilon)$ for an arbitrary constant $\epsilon$.
%
%\begin{enumerate}
%\item The time until the preferred opinion either disappears or prevails is
%$O(\log n/\phi)$, w.h.p.
%\item If the initial number of nodes with preferred opinion is at least $\beta \log n$
%then this opinion prevails w.h.p.
%\end{enumerate}
%\end{theorem}

\begin{theorem}
    \label{thm:ub-bvoter}

Consider the Biased Voter model  in the dynamic regular graph %generation 
model. Assume $\kappa \leq n$  opinions are  arbitrarily distributed over the nodes of $G_1$.
%    Let $\mathcal{G}=G_1, G_2, \dots $ be any sequence of regular dynamic graph with $n$ nodes, $\kappa\le n$ distinct initial opinions,
%where $G_i$ has conductance at least $\phi_i$.
Let $\phi_t$ be a lower bound on the conductance at time step $t$.
Assume that $\alpha_1 \leq 1-\epsilon$, for an arbitrary small constant $\epsilon > 0$.
Assume the initial number of nodes with the preferred opinion is at least $c\log n$, for some constant $c = c(\alpha_1)$.
Then the preferred opinion prevails w.h.p.\ in at most $\tau'''$ steps, where $\tau'''$ is the first round so that
$\sum_{t=1}^{\tau'''} \phi_t \geq b \log n$, for some constant $b$.
    %The  preferred opinion either prevails or vanishes after $\tau'''$ rounds w.h.p..
For static graphs ($G_{i+1}=G_i$ for all $i$), we have w.h.p.\ $T=O(\log n/\phi)$.
\end{theorem}

The assumption on the initial size of the preferred opinion is crucial for the time bound $T=O(\log n/\phi)$, in the sense that there are instances where the  expected consensus time is at least $T=\Omega(n/\phi)$ if the size of the preferred opinion is small.\footnote{Consider a $3$-regular graph and $n$ opinions where all other $\alpha_1 =\alpha_2=\dots=\alpha_{n-1}=1/2$. The preferred opinion vanishes with constant probability and the bound for the standard voter model of Observation~\ref{lem:lowerbound} applies.}

The rumor spreading process can be viewed as an instance of the biased voter model with two opinions having popularity $1$ and $0$, respectively.
However, the techniques used for the analysis of rumor spreading do not  extend to the voter model.
This is due to the fact that rumor spreading is a progressive process, where nodes can change their opinion only once, from ``uninformed'' to ``informed'', whereas they can change their opinions over and over again in the case of the voter model.
Note that the above bound is the same as the bound for rumor spreading of~\cite{G11} (although the latter bound holds for general graphs, rather than just for regular ones).
Hence, our above bound is tight for regular graphs with conductance $\phi$, since the rumor spreading lower bound of
$\Omega(\logn/\phi)$  is also a lower bound for biased voting in our model. %%%%%%%%%%%%%%%%%%%%%
%%%%%%%%%%%%%%%%%%%%%%%%%%%%%%%%%%%%%%%%%%%%%%%%%

%%%%%%%%%%%%%%%%%%%%%%%%%%%%%%%%%%%%%%%%%%%%%%%%%%%%%%%%%%%%%%%%%%%%%%
\section{Analysis of the Voter Model}\label{standardvotermodel}
In this section we show the upper and lower bound for the standard voter model. We begin with some definitions.
Let $G = (V,E)$.
For a fixed set $S\subseteq V$ we define $cut(S,V\setminus S)$ to be the set of edges between the sets $S\subseteq V$ and $V\setminus S$ and let $\lambda_{u}$ be the number of  neighbours of $u$ in $V\setminus S$. Let
 $\vol(S)=\sum_{u\in S} \de_u.$
The \emph{conductance} of $G$ is defined as
%\vspace{-0.1cm}
%$$\phi=\phi(G)=\min \left\{  \tfrac{|cut(U,V \setminus U)|
%}{\vol(U)}\colon U\subset V \text{ with } 0 < \vol(U) \leq \vol(V)/2 \right\}.$$
$$\phi=\phi(G)=\min \left\{  \tfrac{\sum_{u \in U} \lambda_{u}
}{\vol(U)}\colon U\subset V \text{ with } 0 < \vol(U) \leq m\right\}.$$
We note $1/n^2 \leq \phi \leq 1$. %We define the \emph{weighted conductance} of $G$ by
%$$\dmin \phi_t=\dmin \phi_t(G)=\min \left\{  \tfrac{\sum_{u \in U} \lambda_u \de_u
%}{\vol(U)}\colon U\subset V \text{ with } 0 < \vol(U) \leq m \right\}.
%$$
We denote by $v_t^{(i)}$ the set of nodes that have opinion $i$ after the first $t$ rounds and $t\geq0$. If we refer to the random variable we use $V_t^{(i)}$ instead.

First we show  Theorem~\ref{thm:voter} for $\kappa=2$ (two opinions),  which we call $0$ and $1$ in the following.  Then we generalise the result to an arbitrary number of opinions.
We model the system with a Markov chain $M_{t\ge 0}=(V^{(0)}_t,V^{(1)}_t)_{t\ge 0}$.

Let $s
_t$ denote the set having the smaller volume, i.e., $s_t = v^{(0)}_t$ if $\vol(v^{(0)}_t)\leq \vol(v^{(1)}_t)$, and  $s_t = v^{(1)}_t$ otherwise. Note that we use $s_t, v^{(0)}_t$ and $v_t^{(1)}$ whenever the state at time $t$ is fixed, and  $S_t, V^{(0)}_t$ and $V_t^{(1)}$ for the corresponding random 
variables.
For $u\in v^{(0)}_t$, $\lambda_{u,t}$ is  the number of neighbours of $u$ in $V\setminus v^1(t)$ and for $u\in v^{(1)}_t$, $\lambda_{u,t}$ is the number of neighbours of $u$ in $V\setminus v^{(0)}_t$; $d_u$ is the degree of $u$ (the degrees do not change over  time).

To analyse  the process we  use a potential function.
Simply using the volume of nodes sharing the same opinion as the potential function will not work. It is easy to calculate that
the expected volume of nodes with a given opinion does not change in one step.   Instead, we use
a convex function on the number of nodes with the minority opinion.
We define $$\pot{t} = \sqrt{\vol(S_t)}.$$

In Lemma~\ref{lem:technical-st} we first calculate the one-step potential drop of $\pot{t}$.
Then we show that  every opinion either prevails or vanishes once the sum of conductances is proportional to the  volume of nodes having that opinion (see Lemma~\ref{lem:two-op-const-prob}), which we use later to prove Part 1 and 2  of Theorem~\ref{thm:voter}.

%
%Then we derive Theorem~\ref{thm:voter} part 1, i.e., $T\leq \tau$, which bound the consensus time for an arbitrary number of opinions.
%The proof bounds the number of rounds before the number of distinct opinions left halves.
%To prove Theorem~\ref{thm:voter} part 2, i.e., $T\leq \tau'$ we bound the potential drop obtained by Lemma~\ref{lem:technical-st} differently and then apply a drift theorem.
%
%{\bf muss noch neu}

\medskip

\begin{lemma}
    \label{lem:technical-st}
Assume  $s_t\not= \emptyset$ and $\kappa=2$. %if $is
 Then
    \begin{equation*}
        \label{eq:exp-psi}
        \Exp[\pot{t+1} \mid S_t=s_t] \leq \potfixed{t} - \frac{\sum_{u \in V} \lambda_{u,t}\cdot \de_u}{32\cdot (\potfixed{t})^3}.%, \quad\text{if } \potfixed{t}\neq 0.
    \end{equation*}
\end{lemma}%
%\subsection*{Proof of Lemma~\ref{lem:technical-st}}
\begin{proof}
W.l.o.g.\ we assume that opinion $0$ is the minority opinion,
i.e. $0 < \vol(V^{(0)}_t)\leq \vol(V^{(1)}_t)$.
To simplify the notation we omit the index $t$ in this proof and write
$v^{(0)}$ instead of $v_t{(0)}$, $v^{(1)}$ for  $V\setminus v^{(0)}_t$, and
 $\lambda_u$ instead of $\lambda_{u,t}$.
Hence, $s_t = v^{(0)}$ and $\potfixed{t} = \sqrt{\vol(v^{(0)})}$.
Note that for $t=0$ we have $\vol(v^{(0)})=\potfixed{t}^2$.
Furthermore, we fix $S_t=s_t$ in the following (and condition on it). We define $m$ as the number of edges.
Then we have
\begin{align}\label{cat}
    \Exp[\pot{t+1} - \potfixed{t}\mid S_t=s_t]
    &=
    \Exp[\sqrt{\vol(S_{t+1})} - \sqrt{\vol(s_t)}]
  \notag \\&=
    \Exp\left[\sqrt{\min\left\{\vol(V^{(0)}_{t+1} ),m-\vol(V^{(0)}_{t+1} )\right\}  } - \sqrt{\vol(s_t)}\right]
%    \notag\\&
\notag\\&\leq
    \Exp\left[\sqrt{\vol(V^{(0)}_{t+1})} - \sqrt{\vol(v^{(0)})}\right]
\end{align}

\noindent
Now we define
$$
  X_u=\left\{\begin{array}{ll} \de_u & w.p. \ \frac{\lambda_u}{2\cdot\de_u}\
   if\ u\in v^{(1)}\\
-\de_u &  w.p.\ \frac{\lambda_u}{2\cdot\de_u}\ if\ u\in v^{(0)}\\
         0 & otherwise\end{array}
\right. %one dot is eaten by right
$$
and $\Delta=\sum_{u\in V}X_u$.
Note that we have  $\Delta=\vol(V^{(0)}_{t+1}) - \vol(v^{(0)}) $ and

\begin{align*}%\label{stocdom}
    \Exp\left[\sqrt{\vol(V^{(0)}_{t+1})} - \sqrt{\vol(v^{(0)})}\right]
    &=\Exp\left[\sqrt{\vol(v^{(0)})+\Delta} - \sqrt{\vol(v^{(0)})}\right]\\
&=\Exp\left[\sqrt{\vol(v^{(0)})}\left(\sqrt{1+\tfrac{\Delta}{\vol(v^{(0)})}} - 1\right)\right]\\
    &=\potfixed{t}\cdot \Exp[ \sqrt{1+\Delta/\potfixed{t}^2} - 1].
\end{align*}

Unfortunately we cannot bound $\potfixed{t}\cdot \Exp[ \sqrt{1+\Delta/\potfixed{t}^2} - 1]$ directly.
Instead, we define a family of random variables which is closely related to
$X_u$.
$$
     Y_u=\left\{\begin{array}{ll} \lambda_u & w.p.\ \tfrac{1}{2}\mspace{37mu} if\ u\in v^{(1)} \\
-\de_u &  w.p.\ \tfrac{\lambda_u}{2\cdot\de_u}\ \mspace{10mu} if\ u\in v^{(0)}\\\
  0 & otherwise\end{array}\right.
$$
Similarly, we define $\Delta'= \sum_{u\in V} Y(u) $. Note that
$|E[Y_u]|=\lambda_u/2$
for both $u\in v^{(1)}$ and $u\in v^{(0)}$.
%$E[(Y_u)^2]=(\lambda_u)^2/2$ for  $u\in v^{(1)}$ and
%$E[(Y_u)^2]=d_u\cdot \lambda_u/2$ for  $u\in v^{(0)}$.
%
In  Lemma \ref{lem:replaceRV}, we  show that
$$\Exp[ \sqrt{1+\Delta/\potfixed{t}^2} ] \leq \Exp[ \sqrt{1+\Delta'/\potfixed{t}^2}].$$
which results in
$ \Exp[\pot{t+1} - \potfixed{t}\mid S_t=s_t] \leq     \potfixed{t}\cdot \Exp[ \sqrt{1+\Delta'/\potfixed{t}^2} - 1]$
From the Taylor expansion  $\sqrt{1+x} \leq 1+\tfrac x2-\tfrac{x^2}8+\tfrac{x^3}{16}$, $x\geq-1$ it follows that
\begin{eqnarray*}
    \Exp[\pot{t+1} - \potfixed{t}\mid S_t=s_t] &\leq &
    \potfixed{t}\cdot\Exp\big[\tfrac {\Delta'}{2\potfixed{t}^2} -\tfrac{(\Delta')^2}{8\potfixed{t}^4}+\tfrac{(\Delta')^3}{16\potfixed{t}^6}\big].
\end{eqnarray*}
It remains to bound $\Exp[\Delta']$, $\Exp[(\Delta')^2]$, and $\Exp[(\Delta')^3]$.

\begin{itemize}
\item $\Exp[\Delta']$:
We have
$
\Exp[\Delta']  =\sum_{u\in V} E[Y_u] =
 \sum_{u \in v^{(1)}}\frac{\lambda_u}{2} -
 \sum_{v \in v^{(0)}}\frac{\lambda_v}{2}=0,
% &&\sum_{u \in v^{(1)}}\lambda_u - \sum_{v \in v^{(1)}}\lambda_v=0.
$
where the last equality holds since $\sum_{u \in v^{(1)}}\lambda_u$
and $\sum_{u \in v^{(1)}}\lambda_u $ both count the number of edges crossing the cut between $v^{(0)}$ and $v^{(1)}$.

\item $\Exp[(\Delta')^2]$: since
$E[(Y_u)^2]=(\lambda_u)^2/2$ for   $u\in v^{(1)}$ and
$E[(Y_u)^2]=-d_u\cdot \lambda_u/2$ for   $u\in v^{(0)}$
we have
\begin{align}\label{squuuared}
    \Exp[(\Delta')^2]
    &=
    \sum_{u\in V}\Var[Y_u] +(\Exp[Y_u])^2=\sum_{u\in V}\Var[Y_u]+0
    =
    \sum_{u\in V}(\Exp[(Y_u)^2] - (\Exp[Y_u])^2)\notag\\
    &=\sum_{u\in v^{(0)}}(\Exp[(Y_u)^2] - (\Exp[Y_u])^2)+\sum_{u\in v^{(1)}}(\Exp[(Y_u)^2] - (\Exp[Y_u])^2)\notag\\
   & =
    \sum_{u \in v^{(0)}} \frac{\lambda_u d_u}{2}-\sum_{u \in v^{(0)}}  \frac{\lambda_u^2}{4}
+   \sum_{u \in v^{(1)}}  \frac{\lambda_u^2}{4}
%    \tfrac{3}{4d}\sum_{u\in V} \lambda_u
%    =
%    \tfrac{3}{2d}\cdot 2C.
\geq  \sum_{u \in v^{(0)}} \frac{\lambda_u d_u}{4}.%\geq  \tfrac{\dmin}{4}\sum_{u \in V} \lambda_u \geq \tfrac{\dmin \potfixed{t}^2\phi}{2},
\end{align}
%where we used that $\sum_{u \in V} \lambda_u = 2|E(v^{(0)},V^{(1)}(t))| \geq 2\vol(v^{(0)}) \phi$ which follows from the definition of the conductance $\phi$.

\item $\Exp[\Delta'^3]$:
In Lemma~\ref{lem:skewness} we  show that
 $$\Exp[\Delta'^3]= \sum_{u\in V}\big(\Exp[(Y_u)^3]-3\Exp[(Y_u)^2]\cdot\Exp[Y_u] + 2\Exp[Y_u]^3\big).$$
Note that
$E[(Y_u)^3]=\frac12(\lambda_u)^3$ for   $u\in v^{(1)}$ and
$E[(Y_u)^3]=-\frac12\lambda_u\cdot (d_u)^2$ for   $u\in v^{(0)}$.
Hence,
\begin{align}%\label{}
\begin{split}
 \Exp[\Delta'^3] &= \sum_{u \in v^{(0)}}\left(-\frac12{\lambda_u}\cdot (d_u)^2 +\frac34{(\lambda_u)^2}\cdot\de_u -\frac14{\lambda^3_u}\right)  \\
& \phantom{00}+
 \sum_{u \in v^{(1)}}\left( \frac12 (\lambda_u)^3 -  \frac34 (\lambda_u)^3
+\frac14 (\lambda_u)^3\right)\leq 0,
\end{split}
\end{align}
where the first sum is bounded by 0 because  $\lambda_u \leq \de_u$.

\end{itemize}

Combining all the above estimations we get
\begin{align*}
    \Exp[\pot{t+1} - \potfixed{t}\mid S_t=s_t]
    &\leq
    \potfixed{t}\cdot \Exp\left[\frac {\Delta'}{2\potfixed{t}^2} -\frac{\Delta'^2}{8\potfixed{t}^4}+\frac{\Delta'^3}{16\potfixed{t}^6}\right]
%    \leq
%    \potfixed{t} \left( \tfrac 0{2\potfixed{t}^2}
%    -\tfrac{ \sum_{u \in V} \lambda_u\de_u}{32\potfixed{t}^4}
%    +\tfrac{0}{16\potfixed{t}^6}\right)
\leq -\frac{\sum_{u \in v^{(0)}} \lambda_u\de_u}{32\potfixed{t}^3}.
\end{align*}
This completes the proof of Lemma~\ref{lem:technical-st}.
\end{proof}

\subsection{Part 1 of Theorem~\ref{thm:voter}. }
Using Lemma~\ref{lem:technical-st} we show that a given opinion either prevails or vanishes with constant probability as soon as the sum of $\phi_t$  is proportional to the volume of the nodes having that opinion.

\begin{lemma}
    \label{lem:two-op-const-prob}
    Assume that $s_{\hat t}$  is fixed for an arbitrary $(\hat{t}\ge 0)$  and $\kappa=2$.\\
 Let $\tau^*= \min\left\{t' : \sum_{i= \hat t}^{t'}  \phi_i \geq 129 \cdot \vol(s_{\hat t})/\dmin \right\}$.
Then
$\Pr\left(T \leq \tau^*+\hat t\right) \geq 1/2.$\\
In particular, if the graph is static with conductance $\phi$, then $\Pr\big(T \leq \frac{129 \cdot \mbox{vol}(s_{\hat t})}{\phi \cdot \dmin}+\hat t\big) \geq 1/2.$
\end{lemma}
\begin{proof}
From the definition of  $\potfixed{t}$ and $\phi_t$ it follows for all $t$ that
$\potfixed{t}^2=\sum_{u \in v^{(0)}}  \de_u=\mbox{vol$(v^{(0)})$}$ and
$\phi_t\le \sum_{u \in v^{(0)}}\lambda_{u,t}/\mbox{vol$(v^{(0)})$}$.
Hence,
$ \potfixed{t}^2  \cdot \phi_t\cdot \dmin\leq \sum_{u \in v^{(0)}}\lambda_{u,t}\cdot \de_u.   $
Together with
Lemma~\ref{lem:technical-st}  we derive for $s_t\not= \emptyset$
\begin{equation}
        \label{eq:exp-psi2}
        \Exp[\pot{t+1} \mid S_t = s_t  ] \leq
        \potfixed{t} - \frac{\sum_{u \in V} \lambda_{u,t}\cdot \de_u}{32\cdot(\potfixed{t})^3}
        \leq \potfixed{t} - \frac{\dmin\cdot \phi_t}{32\cdot \potfixed{t}}.%, \quad\text{if } \potfixed{t}\neq 0.
    \end{equation}
Recall that $T=\min_t\{ S_t = \emptyset \}$.
In the following we use  the expression $T>t$ to denote the event $s_t\neq \emptyset$.
Using the law of total probability we get
$$\Exp[\pot{t+1} \rvert   T>t]=\Exp\left[\pot{t} - \frac{\dmin\cdot \phi_t}{32\cdot \pot{t}}\Big\rvert  T>t\right]$$
and using Jensen's inequality we get
\begin{align*}
    \Exp[\pot{t+1} \mid  T>t]
    &=\Exp[\pot{t} \mid T>t] - \Exp\left[\frac{\dmin\cdot \phi_t}{32\cdot \pot{t}}\mid T>t\right]
    \\&\leq
    \Exp[\pot{t} \mid T>t] - \frac{\dmin \cdot\phi_t}{32\cdot \Exp[\pot{t}\cdot \mid T>t]}.
\end{align*}
Since $\Exp[\pot{t} \mid  T\leq t]=0$ we have
\begin{eqnarray*}
\Exp[\pot{t}]&=&\Exp[\pot{t} \mid  T>t]\cdot \Pr[T>t]+
\Exp[\pot{t} \mid  T\le t]\cdot \Pr[T\le t]\\
&=& \Exp[\pot{t} \mid  T>t]\cdot \Pr[T>t]+0.
\end{eqnarray*}
Hence,
$$
   \frac{ \Exp[\pot{t+1}]}{\Pr{(T>t)}} \leq
   \frac{\Exp[\pot{t}]}{\Pr{(T>t)}}
 - \frac{\dmin \cdot \phi_t \cdot \Pr{(T>t)}}{32\Exp[\pot{t}]}
$$
and
$$
  \Exp[\pot{t+1}] \leq
 \Exp[\pot{t}]
 - \frac{\dmin \cdot \phi_t \cdot (\Pr{(T>t))^2}}{32\Exp[\pot{t}]}.
$$

Let $t^\ast = \min\{t\colon \Pr(T>t) < 1/2\}$. In the following we use contradiction to show
$$t^\ast \leq \max\{ t:  \sum_{\hat t\leq t<t^\ast}\phi_t \leq 128 \cdot \vol(s_{\hat t})/\dmin \}.$$
Assume the inequality is not satisfied.
With $t = t^\ast-1$ we get
\[
    \Exp[\pot{t^*}] \leq
    \Exp[\pot{t^*-1}] - \frac{\dmin \cdot \phi_t\cdot(\Pr(T>{t^\ast}-1))^2}{32\Exp[\pot{t^*-1}]}
    \leq
    \Exp[\pot{t^*-1}] - \frac{\dmin \cdot\phi_{t^*}
    \cdot(1/4)}{32\Exp[\pot{t^*-1}]}.
\]
Applying this equation iteratively, we obtain
\begin{align}\label{thecontra}
    \Exp[\pot{t^\ast}]
    \leq
    \Exp[\pot{\hat t}] - \sum_{\hat t\leq t<t^\ast}\frac{\dmin \cdot\phi_t\cdot 1/4 }{32\Exp[\pot{t}]}
    \leq
    \Exp[\pot{\hat t}] -  \frac{\dmin \cdot\sum_{\hat t\leq t<t^\ast}\phi_t}{128\Exp[\pot{\hat t}]}.
\end{align}
Using the definition of $\Exp[\pot{\hat t}] = \sqrt{\vol(s_{\hat t})}$ and the definition of $t^\ast$ we get
$$  \Exp[\pot{t^\ast}] <  \sqrt{\vol(s_{\hat t})} -  \frac{\dmin \cdot  128 \cdot \vol(s_{\hat t})}
{128 \cdot \dmin \cdot\sqrt{\vol(s_{\hat t})}}=
\sqrt{\vol(s_{\hat t})} -  \frac{ \vol(s_{\hat t})}
{\sqrt{\vol(s_{\hat t})}} =0.$$
This is a contradiction since  $\Exp[\pot{t^\ast}] $ is non-negative.

From the definition of $t^\ast$, we obtain $\Pr\big(T >  \tau^*+ \hat t \big) < 1/2$,
completing the proof of Lemma~\ref{lem:two-op-const-prob}.
%\ptodo{Was ist hier $\hat{t}$?}
\end{proof}

Now we are ready to show the first part  of the theorem.

\begin{proof}[Proof of Part 1 of Theorem~\ref{thm:voter}]
We  divide the $\tau$ rounds into phases. Phase $i$ starts at time  $\tau_i=\min \{ t: \sum_{j=1}^t \phi_j \geq 2i \} $ for $i\geq 0$ and ends at $\tau_{i+1}-1$. Since $\phi_j \leq 1$ for all $j\geq 0$ we have $\tau_0 < \tau_1 < \dots$ and $\sum_{j=\tau_i}^{\tau_{i+1}} \phi_j \geq 1$ for $i\geq 0$.  Let $\ell_t$ be the number of distinct opinions at the beginning of phase $t$. Hence, $\ell_0=\kappa$.

We  show in Lemma~\ref{fraction} below that the expected number of phases before the number of  opinions drops by a factor of $5/6$ is bounded by
$6c \cdot \vol(V)/(\ell_t\cdot  \dmin)$.
For $i\ge 1$ let $T_i$ be the number of phases needed so that the number of opinions drops to $(5/6)^i\cdot \ell_0$. Then only one opinion remains after
$\log_{6/5}\kappa$ many of these meta-phases.
Then, for a suitably chosen constant $b$,
\begin{eqnarray*}
    \Exp[T]
   & =& \sum_{j=1}^{\log_{6/5}\kappa} E[T_j]\leq
    \sum_{j=1}^{\log_{6/5}\kappa-1} \frac{6c \cdot \vol(V)}{\ell_j \cdot \dmin}
    \leq
    \sum_{j=1}^{\log_{6/5}\kappa}\frac{6c\vol(V)}{(5/6)^{j}\cdot \ell_0\cdot \dmin}=\frac{b \cdot m}{4\cdot \dmin }.
\end{eqnarray*}
By Markov inequality, consensus is reached w.p.\ at least $1/2$ after $b\cdot m/ (2\dmin)$ phases.
By definition of $\tau$ and the definition of the phases, we have that the number of phases up to time step $\tau$ is at least $b\cdot m/ (2\dmin)$. Thus, consensus is reached w.p.\ at least $1/2$ after $\tau$ time steps, which finishes the proof.
\end{proof}

\begin{lemma}\label{fraction}
Fix a phase $t$ and assume $c=129$ and $\ell_t > 1$.  The expected number of phases before the number of  opinions drops to  $5/6\cdot \ell_t$ is bounded by
$6c \cdot \vol(V)/(\ell_t\cdot  \dmin)$.
\end{lemma}
\begin{proof}
Consider a point when there are $\ell'$ opinions left, with $5/6 \cdot \ell< \ell' \leq \ell$.
Among those $\ell'$ opinions, there are at least $\ell'- \ell/3$ opinions $i$ such that the volume of nodes with opinion $i$ is at most $3\cdot \vol(V)/\ell$. Let $S$ denote the set of these opinions and
let $Z_i$ be an indicator variable which is 1 if opinion $i\in S$ vanished after
$s = 3 c\cdot \vol(V)/(\ell\cdot \dmin)$ phases and $Z_i=0$ if it prevails.
To estimate $Z_i$ we consider the process where we have two opinions only.
All nodes with opinion $i$ retain their opinion and all other nodes have opinion $0$. It is easy to see that in both processes the set of nodes with
opinion $i$ remains exactly the same.
Hence, we can use Lemma \ref{lem:two-op-const-prob}
to show that  with probability at least $1/2$, after $s$ phases opinion $i$ either vanishes or prevails.
Hence,  $$\Exp\left[\Sigma_{j\in S}Z_j\right]= \Sigma_{j\in S}\Exp[Z_j]\geq |S|/2 \geq (\ell'- \ell/3)/2 .$$

Using Markov's inequality we get that  with probability $1/2$ at least $(\ell'- \ell/3)/4$ opinions vanish within $s$ phases, and the number of opinions remaining is at most
$\ell' - (\ell'- \ell/3)/4 = 3/4\cdot \ell'+ \ell/12 \leq 5/6\cdot \ell$.
The expected number of phases until $5/6\cdot \ell$ opinions can be bounded by
$\sum_{i=1}^{\infty} 2^{-i}\cdot s\le 2s= \frac{6 c \cdot\vol(V)}{\ell\cdot \dmin}.$
\end{proof}

\subsection{Part 2 of Theorem~\ref{thm:voter}}
The following lemma is similar to Lemma~\ref{lem:two-op-const-prob} in the last section: We first bound the expected potential drop in round $t+1$, i.e., we bound $\Exp[\pot{t+1} -\potfixed{t} \mid S_t=s_t]$.
This time however, we express the drop as a function which is linear in  $\potfixed{t}$. This allows  us to bound the expected size of the potential at time $\tau'$, i.e., $\Exp[\pot{\tau'}]$, directly. From the expected size of the potential at time $\tau'$ we derive the desired bound on  $Pr\left(T \leq \tau' \right)$.

%The proof can be found in the appendix.
\begin{lemma}
    \label{lem:forsecondpart}
    Assume $\kappa=2$. We have
$\Pr\left(T \leq \tau' \right) \geq 1/n^2.$
In particular, if the graph is static with conductance $\phi$, then $\Pr\big(T \leq \frac{96\cdot n\log n}{\phi^2 }\big) \geq 1-1/n^2$.
\end{lemma}
\begin{proof}%[Proof of Lemma~\ref{lem:forsecondpart}]
In the following we fix a point in time $t$ and use $\lambda_u$ instead of $\lambda_{u,t}$.
From Lemma~\ref{lem:technical-st} and  the observation $\lambda_u \leq \de_u$ we obtain for $s_t\not=\emptyset $
$$\Exp[\pot{t+1} -\potfixed{t} \mid S_t=s_t] \le
 - \frac{\sum_{u \in V} \lambda_{u}\cdot \de_u}{32\cdot(\potfixed{t})^3}
         \leq - \frac{\sum_{v\in v^{(0)}} (\lambda_u)^2 }{32(\potfixed{t})^3}.$$
We have, by Cauchy-Schwarz inequality,
$\left(\sum_{v\in v^{(0)}} \lambda_u\right)^2= \left(\sum_{v\in v^{(0)}} \lambda_u\cdot 1 \right)^2 \leq \sum_{v\in v^{(0)}} (\lambda_u)^2 \cdot n. $
Hence,
$$\sum_{v\in V} (\lambda_u)^2 \geq \sum_{v\in v^{(0)}} (\lambda_u)^2 \geq \frac{\left(\sum_{v\in v^{(0)}} \lambda_u\right)^2}{n} \geq
\frac{(\mbox{vol}(v^0))^2\cdot(\phi_t)^2}{ n}
\geq \frac{ \potfixed{t}^4 \cdot (\phi_t)^2}{n},$$ where the third inequality follows by definition of $\phi_t$.
Hence, for $s_t\not=\emptyset $ we have
$$\Exp[\pot{t+1} -\potfixed{t} \mid  S_t=s_t] \leq - \frac{ \potfixed{t} \cdot(\phi_t)^2}{32n}.$$
Note that $\Exp[\pot{t+1} \mid  S_t=\emptyset ]  = 0=(1- \frac{ (\phi_t)^2}{32n})\cdot \Exp[\Psi(\emptyset)].$
Hence, for all $t\geq 1$ we get $$\Exp[\pot{t+1}]=\Exp[\Exp[\pot{t+1} \mid  S_t=s_t ]] \leq \left(1- \frac{ (\phi_t)^2}{32n}\right)\cdot E[\potfixed{t}].$$
Applying this recursively yields
\begin{align*} \Exp[\pot{t+1}]&\leq \pot{0} \cdot\prod_{i=0}^t (1- \frac{ (\phi_i)^2}{32n}) 
\leq  \pot{0} \cdot\left(1- \frac{1}{t+1}\sum_{i\leq t} \frac{ (\phi_i)^2}{ 32n}\right)^{t+1}
\leq   \pot{0} \cdot \exp\left(\sum_{i\leq t} \frac{ (\phi_i)^2}{ 32n}\right),\\&
\end{align*}
where the second inequality follows from the Inequality of arithmetic and geometric means.

By definition of $\tau'$, and from the observation $\pot{0}\leq n$ we get that $\Exp[\pot{\tau'}]\leq n^{-2}.$

%We are interested in $\Pr(\pot{\tau'}=0)$, since $\Pr(\pot{\tau'}=0)=\Pr( T\leq \tau' ).$

We derive
\begin{align}\label{blue}
n^{-2 }\geq  \Exp[\pot{\tau'}] \geq 0\cdot\Pr(\pot{\tau'}=0) + 1 \cdot (1-\Pr(\pot{\tau'}=0)),
\end{align}
where we used that $\min \{\Psi(S) : S\subseteq V: S \not=\emptyset\} \geq 1$.
Solving \eqref{blue} for $\Pr(\pot{\tau'}=0$ gives
 $\Pr(\pot{\tau'}=0)\geq 1-1/n^2$.

Since $\kappa=2$, it follows that $\Pr( T\leq \tau' )=\Pr(\pot{\tau'}=0)\geq 1-1/n^2$, which yields the claim.
\end{proof}

We now prove Part 2 of Theorem~\ref{thm:voter} which generalises to $\kappa > 2$.
\begin{proof}[Proof of Part 2 of Theorem~\ref{thm:voter}]
We define a parameterized version of the consensus time $T$. We define $T(\kappa) =\min\{t\colon \pot{t} = 0 : \text{the number of different opinions at time $t$ is $\kappa$}\}$ for $\kappa \leq n$.
We want to show that $\Pr( T(\kappa)\leq \tau' )\geq 1-1/n$.
From Lemma~\ref{lem:forsecondpart} we have that, that $\Pr( T(2)\leq \tau' )\geq 1-1/n^2$.
We define the 0/1 random variable $Z_i$ to be one if opinion $i$ vanishes or is the only remaining
opinion after $\tau'$ rounds and $Z_i=0$ otherwise. We have that $\Pr(Z_i = 1 )\geq 1-1/n^2$ for all $i\leq \kappa$.
We derive $\Pr( T(\kappa)\leq \tau' ) = \Pr( \land_{i\leq \kappa} Z_i) \geq 1-1/n$, by union bound. This yields the claim.
\end{proof}

\subsection{Lower Bounds}%Optimality of the Bounds of Theorem~\ref{thm:voter}}

%
%
%\section{Remarks}

%In the following we state a drift theorem, which can be found, e.g., in \cite{DJW12}.
%\begin{theorem}
%\label{drift}
%Let $S \subseteq \mathbb{R}$ be a finite set of positive numbers with minimum $s_{min}$. Let $\{V(t)\}_{t \in \mathbb{N}}$ be a sequence of random variables over $S \cup \{0\}$. Let $T$ be the random variable that denotes the first point in time $t \in \mathbb{N}$ for which $V(t) = 0$.
%Suppose that there exists a constant $\delta > 0$ such that
%$
%E[{V(t) - V(t+1)\ |\ V(t) = s}] \geq \delta s
%$
%holds for all $s \in S$ with $Pr({V(t) = s}) > 0$. Then for all $s_0 \in S$ with
%$Pr({V(0) = s_0}) > 0$,
%$
%E[{T\ |\ V(0) = s_0}] \leq \tfrac{1 + \ln(s_0/s_{min})}{\delta}.
%$\end{theorem}

In this section, we give the intuition behind the proof of Theorem~\ref{st:adaptivelowerbound} and state two additional observations.
Recall that Theorem~\ref{st:adaptivelowerbound}  shows that
our bound for regular graphs is tight for the  adaptive adversary, even for $k=2$.
The first observation shows that the expected consensus time can be super-exponential
 if the adversary is allowed to change the degree sequence.
The second observation can be regarded as a (weaker)  counter part of  Theorem~\ref{st:adaptivelowerbound} showing a lower bound of $\Omega(n/\phi)$ for static graphs, assuming that either $d$ or $\phi$ is constant. 

We now give the intuition behind the proof of Theorem~\ref{st:adaptivelowerbound}. %and refer the reader to the Section \ref{sec:lower} in the appendix for the full proof.
The high level approach is as follows.
For every step $t$ we define an adaptive adversary that  chooses $G_{t+1}$ after observing $V_{t}^{(0)}$ and $V_{t}^{(1)}$.
The adversary chooses  $G_{t+1}$ such that the cut between $V_{t}^{(0)}$ and $V_{t}^{(1)}$  is of order of $\Theta(\phi_t \cdot d n)$.  We show that such a graph exists when the number of nodes in both  $V_{t}^{(0)}$ and $V_{t}^{(1)}$ is at least of linear size (in $n$).
By this choice the adversary ensures that the expected potential drop of $\pot{t+1}$ at most $-c\phi_t d/\potfixed{t} $ for some constant $c$.
Then we use the expected potential drop, together with the optional stopping theorem, to derive our lower bound.

We proceed by giving a lower bound on the potential drop assuming a cut-size of $\Theta(\phi_t \cdot d\cdot n)$.

\begin{lemma}\label{claiminlower}
Assume $|cut(s_t, V\setminus s_t)|
\leq c\phi_t \cdot d n$ for some constant $c$. %There exists a constant $c$.\ such that
 %if $T>t$ then the conditional expectation of $\hat\pot{t+1}$, given $\hat\Psi_{t}$, is
   Then we have
    \begin{equation}\label{eq:potdroplower}
        \Exp[\pot{t+1} \mid S_t = s_t] \geq \potfixed{t} -\frac{c\cdot \phi_t \cdot d}{\potfixed{t}} .%, \quad\text{if } \hat\potfixed{t}\neq 0.
    \end{equation}
\end{lemma}
%Let $V^{(0)}(t)=v^{(0)}$,  $v^{(1)} = V\setminus V^{(0)}(t)$, and assume w.l.o.g.\ that $0 < \vol(v^{(0)})\leq m$; then $s_t = v^{(0)}$ and $\hat\potfixed{t} = {\vol(v^{(0)})^2}$.
%We have $\vol(S_{t+1})=v^{(0)} + \Delta$ with $\Delta=\sum_u X_u$, where
%$$
%     X_u=\left\{\begin{array}{ll} \de_u & w.p. \ \tfrac{\lambda_u}{2\cdot\de_u}\ if\ u\in v^{(1)}\\
%-\de(u) &  w.p.\ \tfrac{\lambda_u}{2\cdot\de_u}\ if\ u\in v^{(0)}\\
%         0 & otherwise\end{array}\right. . %one dot is eaten by right
\begin{proof}
The proof is similar to the proof of Lemma \ref{lem:technical-st}.
\begin{align*}
  \Exp&\left[\pot{t+1} -\potfixed{t} \mid S_t = s_t\right] =\\
&=\Exp[\sqrt{\vol(S_{t+1})} - \sqrt{\vol(s_t)}]
   \\&=
%\Exp\left[\min \left\{\sqrt{\vol(v^{(0)}) + \Delta } ,\sqrt{\vol(v^{(1)})-\Delta}\right\}-  \potfixed{t} \right]\\
  \Exp\left[\sqrt{\min\left\{\vol(V^{(0)}_{t+1} ),m-\vol(V^{(0)}_{t+1} )\right\}  } - \sqrt{\vol(s_t)}\right]\\
&=
\Exp\left[\min \left\{\sqrt{\vol(v^{(0)}_{t} ) + \Delta } ,\sqrt{m-(\vol(v^{(0)}_{t} )+\Delta)}\right\}-   \sqrt{\vol(v^{(0)}_{t} )} \right]\\
&\geq
\Exp\left[\min \left\{\sqrt{\vol(v^{(0)}_{t} ) + \Delta } ,\sqrt{\vol(v^{(0)}_{t} )-\Delta}\right\}-   \sqrt{\vol(v^{(0)}_{t} )} \right]\\
 &=
\Exp\left[\sqrt{\vol(v^{(0)}_{t} ) - |\Delta|} -  \sqrt{\vol(v^{(0)}_{t} )}\right]\\
 &=
\Exp\left[ \sqrt{\vol(v^{(0)}_{t} )}\left(\sqrt{1 - \tfrac{|\Delta|}{ \vol(v^{(0)}_{t} )}} -  1\right) \right]\\
 &\geq
\potfixed{t}\cdot \Exp\left[ \frac{|\Delta|}{2\potfixed{t}^2}-\frac{|\Delta|^2}{\potfixed{t}^4} \right] \numberthis\label{huhu}
\end{align*}
where the last inequality comes from the Taylor expansion inequality $\sqrt{1+x} \geq 1+\frac x2-x^2$, $x\geq-1$.

\noindent
Similar to \eqref{squuuared} we get
\begin{align*}
    \Exp[(\Delta)^2]
    &=
    \sum_{u\in V}\Var[X_u] +(\Exp[X_u])^2=\sum_{u\in V}\Var[X_u]+0
    =
    \sum_{u\in V}(\Exp[X_u^2] - (\Exp[X_u])^2)\\
 &\leq
    \sum_{u\in V}\Exp[X_u^2]
    =
    \sum_{u \in V} \frac{\lambda_u \cdot d}{2} = d\cdot |\mbox{cut}(s_t, V\setminus s_t)|  \leq c \cdot \phi_t \cdot d\cdot \vol(s_t)= c \cdot\phi_t\cdot d\potfixed{t}^2.
\end{align*}
From \eqref{huhu} we derive now
\begin{align*}
  \Exp\left[\pot{t+1} -\potfixed{t} \mid S_t=s_t\right]
 &\geq
\potfixed{t} \cdot \left( -\frac{c\cdot \phi_t \cdot d}{\potfixed{t}^2} \right)\geq -\frac{c\cdot \phi_t \cdot d}{\potfixed{t}}.
\end{align*}

\end{proof}

%\subsection{Proof of Lemma~\ref{st:adaptivelowerbound}}

%\begin{proof}[Proof of Theorem~\ref{st:adaptivelowerbound}]
%The main idea is to define an adversary that creates in every round  a graph minimizing the cut between  nodes of different opinions  (Lemma~\ref{lem:givemecut}). We then derive a potential drop at time $t$ which is linear in $\phi_t$ allowing us to use a martingale argument showing that $\Pr(T \geq \tau'')\geq 1/2$ which yields the claim.
%
%For the proof we use the same notation as in Lemma~\ref{lem:technical-st}.

We now describe the adversary for the graphs that we use in our lower bound. We assume that we have initially two opinions and each of the two opinion is on $n/2$ nodes.Recall that we can assume that in round $t+1$ the adversary knows the graph $G_t$ as well as the distribution of the opinions over the nodes.
The adversary generates $G_{t+1}$ as follows. $c$ is a constant which is defined in Lemma~\ref{lem:givemecut}).
\begin{itemize}
\item If $|s_t|\geq \gamma \cdot n$, the adversary creates a $d$-regular graph $G_{t+1}$ with two subsets $s_t$ and $V\setminus s_t$ such that the conductance of the $cut(s_t, V\setminus s_t)$ is at most $c\cdot \phi_t$;
According to  Lemma~\ref{lem:givemecut} such a graph always exist.
\item If $|S_t|= |s_t|< \gamma \cdot n$, the adversary does not change the graph and sets $G_{t+1}=G_{t}$.
\end{itemize}

\medskip

To show Theorem~\ref{st:adaptivelowerbound} we  first define a new potential function $g$ and
bound the one step potential drop.  For $x\geq 0$ we
define $g(x)=   \frac{x^2}{2c d}.$
Since $g(\cdot)$ is convex we obtain from Lemma~\ref{claiminlower}, together with  by Jensen's inequality that

\begin{align*}
  \Exp\left[g\left(\pot{t+1}\right) -g\left(\potfixed{t}\right) \mid S_t=s_t\right] &\geq g\left(  \Exp\left[\pot{t+1}\mid S_t = s_t\right]\right) - g\left(\potfixed{t}\right)\\
&\geq  \frac{1}{2c d} \cdot \left(\left(\potfixed{t} - \frac{c\cdot \phi_t \cdot d}{\potfixed{t}}\right)^2-(\potfixed{t})^2\right) \\
&\geq -\phi_t
.
\end{align*}
In the following lemma  we  use  some Martingale arguments to show that with a probability of at least one half $ \vol(S_{\tau''}) \geq \gamma d\cdot n/2$. This implies that no opinion vanished after $\tau''$ w.p. $1/2$, which yields Theorem~\ref{st:adaptivelowerbound}.

\begin{lemma} Let $|s_0|=n/2$. Fix some constant $\gamma < 1/4$.
Assume $G_1,G_2,\ldots$ is s sequence of graphs generated by the adversary defined above. Then
$$P\left(g(\pot{\tau''})>  \frac{2\gamma\cdot n}{4c}\right)\ge\frac12.$$
\end{lemma}
\begin{proof}

%Let $(\Delta^*_t|S_t=s_t)$ be the change in the potential at time $t$, i.e, $(\Delta^*_t|S_t=s_t)=g(\pot{t+1})-g(\potfixed{t})  $.
We define $T'=\min \{ \tau'', \min \{  t: vol(S_t) \leq 2\gamma d n/2\}\}$, where we recall that
 $\tau''$ is the first rounds so that $\sum_{t=1}^{\tau''} \phi_t \geq b n.
$. We define  $$Z_t=g(\pot{t})  + \sum_{i\leq t}\phi_i$$ and show it the following that $Z_{t\land T' }$ is a sub-martingale with respect to the sequence $S_1, S_2, \dots$, where $t\land T' = \min \{ t,T' \}$.

\paragraph*{\bf Case $t < T'$:}
For any $t < T'$ we have
\begin{align}\label{martin}
E[Z_{t+1\land T' }\ |\ S_{t},\dots, S_1] &=
E[Z_{t+1}\ |\ S_{t},\dots, S_1] \notag\\&= \Exp\left[g(\pot{t+1})
+ \sum_{i\leq t+1}\phi_i\  \middle |\  S_{t},\dots, S_1\right]\notag\\
&= \Exp\left[ g(\pot{t+1})-g(\potfixed{t}) \   \middle|\ S_{t},\dots, S_1\right] + g(\potfixed{t}) + \sum_{i\leq t+1}\phi_i\notag\\
&\geq  -\phi_{t+1}+ \left( g(\potfixed{t}) +  \sum_{i\leq t}\phi_i  \right) + \phi_{t+1}  \notag\\
&= Z_{t}=Z_{t\land T'},
\end{align}
where the last equality follows since $t< T'$.
%For $t  =  T'$ we have, by Lemma~\ref{lem:givemecut} and by construction of $G_{t}$, that $\pot{T'}\geq \gamma d n /2$. Thus, similar to \eqref{martin}
%we get
%\begin{align}
%E[Z_{t+1\land T' } | S_{t\land T'},\dots, S_1] &=
%E[Z_{t+1} | S_{t},\dots, S_1] \notag\\&= \Exp\left[g(\pot{t+1})  + \sum_{i\leq t+1}\phi_i \middle| S_{t},\dots, S_1\right]\notag\\
%&= \Exp\left[\Delta^*_{t+1}   \middle| Z_{t},\dots, Z_1\right] + g(\pot{t}) + \sum_{i\leq t+1}\phi_i\notag\\
%&\geq  -\phi_{t+1}+ \left( g(\pot{t}) +  \sum_{i\leq t}\phi_i  \right) + \phi_{t+1}  \notag\\
%&= Z_{t}=Z_{t\land T'}.
%\end{align}
%
%holds and we have $Z_{t+1\land T'} = Z_{T'}$ get
%$E[Z_{t+1\land T'} | S_{t\land T'},\dots, S_1] \geq Z_{t\land T'}$.

\paragraph*{\bf Case $t  \geq  T'$:}
For $t  \geq  T'$   and we have $$E[Z_{(t+1)\land T'}\ |\ S_{t},\dots, S_1] = E[Z_{T'}\ |\ S_{t},\dots, S_1]  =Z_{T'}=Z_{t\land T'},$$
where the last equality follows since $t \geq T'$.
Both cases together show that  $Z_{t\land T' }$ is a sub-martingale. %
According to Theorem~\ref{st:adaptivelowerbound}  we have $T' < \infty$.
%We define the stopping time $T^* = \min\{T', \tau'' \}$ and we have $T^*\leq T'\leq \infty$.
%Since $Z_t$ is a sub-martingale for all $t\geq 0$ we have that $Z'_{t}=Z_{\min\{t,T^*\}}$ is a sub-martingale for all $t\geq 0$.
Hence we can apply the optional stopping-time Theorem (c.f. Theorem~\ref{optional}), which results in  $$E[Z_{t\land T' }]\geq E[Z_0]=\frac{n}{4c}.$$
\noindent
We define $$p=P\left(g(\pot{T'})\leq \frac{2\gamma\cdot n}{4 c}\right).$$
%Thus $E[Z_{T'}]\leq \frac{2\gamma\cdot n}{ 4c} +  \sum_{i\leq T'}\phi_i .$
By definition of $S_t$, we have $|S_t| \leq s_0$ and thus
$$g(\pot{T'}) \leq g(\potfixed{0})=\frac{n}{4c}.$$
Thus $E[Z_{\tau''}]\leq \frac{n}{4c} +  \sum_{i\leq \tau''}\phi_i$.
Hence we derive using $T'\leq \tau''$ that
\begin{align*}
\frac{n}{4c}&=E[Z_0]\leq E[Z_{T'}] \\&\leq p \cdot\left( \frac{2\gamma\cdot n}{ 4c} +  \sum_{i\leq T'}\phi_i  \right)+ (1-p)\cdot \left( \frac{n}{4c} +  \sum_{i\leq \tau''}\phi_i  \right)  \notag\\
&\leq 2 p\cdot  \gamma \cdot \frac{n}{4c} + (1-p)\cdot \frac{n}{4c}  +  \sum_{i\leq \tau''}\phi_i \notag\\
&\leq    2 p\cdot \gamma\cdot \frac{n}{4c}+ (1-p)\cdot\frac{n}{4c}  +  b \cdot n+1,
\end{align*}

where the last inequality follows from the definition of $\tau''$ together with the fact that $\phi_i\leq 1$ for all $i$.
Hence
$$ 0\leq 2p \cdot \gamma  \cdot \frac{n}{4c} -p\cdot \frac{n}{4c}  +  b \cdot n+1 $$ which equals
$$p\cdot ( 1-2\gamma) \cdot \frac{n}{4c} \leq b\cdot n + 1.$$
Thus
 for $\gamma<1/4$ and $b<1/(18 c)$ we get $p\le 1/2$ and thus we have $P\left(g(\pot{T'})\leq \frac{2\gamma\cdot n}{4 c}\right)=p\le 1/2$ which yields the claim.
\end{proof}

In the following we observe that if the adversary is allowed to change the degrees, then
the expected consensus time is super-exponential.

\begin{observation}[super exponential runtime]\label{dontlettheadversarychangethedegrees}
Let $G_1=(V,E_1),\ G_2=(V,E_2),\ \dots$ be a sequence of graphs with  $n$ nodes,  where the edges $E_1, E_2, \dots$ are distributed by an adaptive adversary,
then the expected consensus time  is at least $\Omega((\sfrac{n}{c})^{\sfrac{n}{c}})$ for some constant $c$.
\end{observation}

\begin{proof}%[Proof Sketch of Observation~\ref{dontlettheadversarychangethedegrees}]
The idea is the following. The initial network is a line and there are two opinions $0$ and $1$ distributed on first $n/2$ and last $n/2$ nodes respectively. Whenever one opinion $i$ has fewer nodes than the other, the adversary creates a graph where only one edge is crossing the cut between both opinions and the smaller opinions forms a clique.
Hence, the probability for the smaller opinion to decrease is $O(1/n)$ and the probability for the bigger opinion to decrease is $\Omega(1)$. This can be coupled with  a biased random walk on a line of $n/4$ nodes with the a transition probability  at most $ 1/(n/4)=4/n$ in one direction and at least $1/2$ in the other. Consequently, there exists a constant $c$ such that expected consensus time is $\Omega((\sfrac{n}{c})^{\sfrac{n}{c}})$.
\end{proof}

The following observation shows that the bound of Theorem \ref{thm:voter} for static regular graphs of $O(n/\phi)$ is tight for regular graphs if either the degree or the conductance is constant. 

\begin{observation}[lower bound static graph]\label{lem:lowerbound}
For every $n$, $d\geq 3$, and constant $\phi$, there exists
a $d$-regular graph $G$ with $n$ nodes and a constant conductance such that the expected consensus time on $G$ is $\Omega(n ) $.
Furthermore, for every even $n$, $\phi>1/n$, and constant $d$, there exists a (static) $d$-regular graph $G$ with $\Theta(n)$ nodes and a conductance of $\Theta(\phi)$ such that the expected consensus time on $G$ is $\Omega(n/\phi ) $.
\end{observation}

\begin{proof}%[Proof Sketch of Observation~\ref{lem:lowerbound}]
%\subsection{Proof Sketch Observation \ref{lem:lowerbound}}

For the first part of the claim we argue that
the meeting time of a $d$-regular graphs is $\Omega(n)$~\cite{AF14}. The claim follows from the duality between coalescing random walks and the voter model.
For the second part of the claim we construct a random graph $G'$ with $n'=\Theta(n\cdot \phi)$ nodes and a constant conductance (See, e.g.,~\cite{ASS08}).
We obtain $G$ by replacing every edge $(u,v)$ of $G'$ with a path connecting $u$ and $v$ of length
 $\ell=\Theta(1/\phi)$ and $\ell \mod d = 0$.
Additionally, we make $G$ $d$-regular by adding  $\ell(d-2)/d$ nodes to every path in such a way that the distance between $u$ and $v$ is maximised.
We note that the obtained graph $G$  has  conductance  $\Theta(\phi)$, $\Theta(n)$ nodes, and is $d$-regular.
The meeting time of $G'$ is $\Omega(n')$~\cite{AF14}.
Taking into account that traversing every path in $G$ takes in expectation $\ell^2=\Theta(1/\phi^2)$ rounds, the expected meeting time in $G$  is at least $\Omega(n' \cdot \ell^2)=\Omega(n/\phi)$.
This yields the claim.
\end{proof}

%
%The following observation shows that the bound of Theorem \ref{thm:voter} for static regular graphs of $O(n/\phi)$ is tight for regular graphs if either the degree or the conductance is constant.
%
%\begin{observation}[lower bound static graph]\label{lem:lowerbound}
%For every $n$, $d\geq 3$, and constant $\phi$, there exists
%a $d$-regular graph $G$ with $n$ nodes and a constant conductance such that the expected consensus time on $G$ is $\Omega(n ) $.
%Furthermore, for every even $n$, $\phi>1/n$, and constant $d$, there exists a (static) $d$-regular graph $G$ with $\Theta(n)$ nodes and a conductance of $\Theta(\phi)$ such that the expected consensus time on $G$ is $\Omega(n/\phi ) $.
%\end{observation}

\section{Analysis of the Biased Voter Model}\label{biasedvotermodel}

In this section, we prove Theorem~\ref{thm:ub-bvoter}.
We show that the set $S_t$ of nodes with the  preferred opinion grows roughly at a rate of $1+\Theta(\phi_t)$, as long as
$S_t$ or $S'_t$  has at least logarithmic size.
For the analysis we break each round down into several steps, where exactly one node which  has at least one neighbour in the opposite set is considered.
Instead of analysing the  growth of $S_t$   for every round
we consider larger time {\em intervals}
consisting of  a suitably chosen number of steps. We change the process slightly by assuming that there is always one node with the preferred opinion.
If in some round  the preferred opinions vanishes totally, Node $1$ is set back to the preferred opinion. Symmetrically, if all other opinions vanishes, then Node $1$ is set to Opinion $1$. Note that this will only increase the runtime of the process.

\medskip

The proof unfolds in the following way.
First, we define formally the \emph{step sequence} $\mathcal{S}$.
Second, we define (Definition~\ref{def})  a step sequence $\mathcal{S}$ to be \emph{good} if, intuitively speaking,
 the preferred opinion grows quickly enough in any sufficiently large subsequence of $\mathcal{S}$.
Afterward, we show that if $\mathcal{S}$ is a good step sequence, then the preferred opinion prevails in at most $\tau'''$ rounds (Lemma~\ref{seqwhp}).
Finally, we show that that  $\mathcal{S}$ is indeed a good step sequence w.h.p. (Lemma~\ref{seqwin}).

We now give some definitions.
Again, we denote by $S_t$ the random set of nodes that have the preferred opinion right after the first $t$ rounds, and let $S'_t = V\setminus S_t$. For a fixed time step $t$ we  write $s_t$ and $s'_t$.
We define the {\em boundary} $\partial s_t $  as the subset  of nodes in $s'_t$ which are adjacent to at least one node from $s_t$. We use the symmetric definition for $\partial s'_t$.
For each $u\in V$, let $\lambda_{u,t}$ be the number of edges incident with $u$ crossing the cut $cut(s_t,s'_t)$, or equivalently, the number of $u$'s neighbours that have a different opinion than $u$'s before round $t$.

We divide each round $t$ into $|s_t| + |s'_t|$ \emph{steps}, in every step a single node $v$ from either $\partial s_t$ or $\partial s'_t$
randomly chooses a neighbour $u$ and  adopts its opinion with the corresponding bias.
Note that we assume that $v$ sees $u$'s opinion referring to \emph{beginning} of the round, even if
was considered before $v$ and changed its opinion in the meantime.
It is convenient to label the steps independently of the round
in which they take place. Hence, step $i$  denotes the $i$-th step counted from the \emph{beginning} of the first round.
Also $u_i$ refers to the node which considered in  step $i$ and $\lambda_i = \lambda_{u_i,t}$. We define
the indicator variable $o_i$ with  $o_i = 1$ if $u_i$ has the preferred opinion and $o_i = 0$ otherwise.
Let
$$\Lambda(i) = \sum_{j=1}^i (1-o_i)\cdot \lambda_i\qquad  {\mbox and}\qquad
\Lambda'(i) = \sum_{j=1}^i o_i\cdot \lambda_i.$$

Unfortunately, the order in which the nodes are considered in a round is important for our analysis and cannot be arbitrarily.
Intuitively, we  order the nodes  in $s_t$ and $s'_t$ such that sum of the degrees of nodes
which are already considered from $s_t$ and the sum of the degrees of nodes already considered from $s'_t$ differs
by at most $d$, i.e.,
\begin{align}
\label{gapLambda}
\phantom{why the hell does the style move it so far }|\Lambda_i-\Lambda'_i|\leq d.
\end{align}

The following rule determines the node to be considered in step $j+1$:
if $\Lambda(j)\leq \Lambda'(j)$, then the (not jet considered) node $v\in \partial s_t$ is with smallest identifier is considered. Otherwise the node $v\in \partial s_t$ is with smallest identifier is considered. Note that at the first step $i$ of any round we have $\Lambda_i=\Lambda'_i$.
This guarantees that \eqref{gapLambda} holds. The \emph{step sequence} $\mathcal{S}$ is now defined as a sequence of tuples, i.e.,  $\mathcal{S}= (u_1, Z_1),\ 
(u_2, Z_2),\ \dots $,  where   $Z_j=1$  if $u_j$ changed its opinion in step $j$ and $Z_j=0$ otherwise for all $j\geq 1$.
Observe that when  given the initial assignment and  the sequence up to step $i$, then we know the \emph{configuration} $\mathcal{C}_i$ of the system, i.e., the opinions of all nodes  at step $i$ and in which round step $i$ occurred.
\medskip

%\subsection{Analysis of the Biased Voter Model}
In our analysis we  consider the increase in the number of nodes with the preferred opinion
in time intervals which contain a sufficiently large number of steps, instead of considering one
round after the other. The following definitions define these intervals.

For all $i, k \geq 0$ where $\mathcal{C}_i$ is fixed, we define the random variable $S_{i,k} := \min\{j\colon \Lambda_j - \Lambda_i \geq k\}$, which is the first time step such that nodes with a degree-sum  of at least $k$ were considered. Let $I_{i,k}= [i+1, S_{i,k}]$ be the corresponding interval where we note that the length is a random variable.
We proceed by showing an easy observation.

\begin{observation}\label{Ishort}
The number of steps in the interval $I_{i,k}$ is at most $2k+2d$, i.e., $|I_{i,k}|\le 2k+2d$.
Furthermore,
$\Lambda'(S_{i,k})-\Lambda'(i) \leq  k+2d.$
\end{observation}

\begin{proof}%[Proof of Observation~\ref{Ishort}]

%Consider the following proof of contradiction.
Assume, for the sake of contradiction, that
$|I_{i,k}| > 2k+2d$. This implies, by definition of $I_{i,k}=[i+1,S_{i,k}]$,  that
 the number of edges of the preferred opinion considered in $I=[i+1,2k+2d]$ is strictly less than $k$, i.e., $\Lambda(i+2k+2d)-\Lambda(i) <  k.$
%Assume $\Lambda'(i+2k+d)-\Lambda'(i)< k.$
% We have that
%$$\Lambda(i+k)-\Lambda(i) =\sum_{j=i}^{i+k} \lambda_{L[j]}\cdot\mathbf{1}_{L[j]\in s'_{round(j)}}\geq \sum_{j=i}^{i+k} \mathbf{1}_{L[j]\in s'_{round(j)}} > k-\tfrac{(k-2d)}{2}.$$
We have
\begin{align}%\label{}
\begin{split}
  \Lambda'(i+2k+2d)-\Lambda(i+2k+2d)&>\Lambda'(i+2k+2d) - (\Lambda(i)+k) \\
&\geq \Lambda'(i)+ 2k+2d -k - (\Lambda(i)+k)\\
&=\Lambda'(i)-\Lambda(i)+2d\\
&\geq -d + 2d \geq d,
\end{split}
\end{align}
 which contradicts \eqref{gapLambda}.
Hence $|I_{i,k}|\le 2k+2d.$

We now prove the second part of the lemma.
Assume, for the sake of contradiction, $\Lambda'(S_{i,k})-\Lambda'(i) >  k+2d.$ This implies that
\begin{align}%\label{}
\begin{split}
  \Lambda'(S_{i,k})-\Lambda(S_{i,k})
%&>\Lambda(i+2k+2d) - (\Lambda'(i)+k) \\
&> \Lambda'(i)+ k+2d -\Lambda(S_{i,k})\\
&\geq \Lambda'(i)+ k+2d - (\Lambda(i)+k)\\
&=\Lambda'(i)-\Lambda(i)+2d\\
&\geq -d + 2d \geq d,
\end{split}
\end{align}
 which contradicts \eqref{gapLambda}.
Hence $\Lambda'(S_{i,k})-\Lambda'(i) \leq k + 2d.$
\end{proof}

%\subsection{Proof of \autoref{} }

Fix $\mathcal{C}_i$ and let $X_{i,k}$ be the total number of times during interval $I_{i,k}$ that a switch from a non-preferred opinion to the preferred one occurs; and define $X'_{i,k}$ similarly for the reverse switches.
Finally, we define $Y_{i,k} = X_{i,k} - X'_{i,k}$; thus $Y_{i,k}$ is the increase in number of nodes that have the preferred opinion during the time interval $I_{i,k}$.
%We now show our main result.
%\subsection{Proof of Theorem~\ref{thm:ub-bvoter}}
%\begin{proof}[Proof of Theorem~\ref{thm:ub-bvoter}]

Define $\ell= \frac{132\beta\log n}{(1-\alpha_1)^2}$ and
$\beta'= \frac{600d}{\alpha_1\cdot (1-\alpha_1)^2}$.
In the following we define a \emph{good} sequence.

\begin{definition}\label{def}
We call the sequence $\mathcal{S}$ of steps \emph{good} if it has all
of the following properties for all $i\leq T'=  2\beta' \cdot n$.
Consider the first $T'$ steps of $\mathcal{S}$ (fix $\mathcal{C}_{T'}$). Then,
\begin{enumerate}[(a)]
\item  $Y_{0,T'} \geq 2n$. (The preferred opinion prevails in at most $T'$ steps)
\item   $Y_{0,i} +|S_0|  > 1$. (The preferred opinion never vanishes)
\item For any $1\leq k \leq T'$ we have $Y_{i,k}\geq  -\ell$. (\# nodes of the pref. opinion never drops by  $\ell$)
\item For any $\ell \leq \gamma \leq T'$,  we have
 $Y_{i,k} >\gamma$, where $k=\gamma \cdot \beta' $. (\# nodes of the pref. opinion increases)

\end{enumerate}

\end{definition}

This definition allows us to prove in a convenient way that a  step sequence $\mathcal{S}$ is w.h.p.\ good: For each property, we simply consider each (sufficiently large) subsequence $S$ separately and we show that  w.h.p. $S$ has the desired property.  We achieve this by using a  concentration bound on $Y_{i,k}$ which we establish in Lemma~\ref{lem:boundsY}. Afterward, we take union bound over all of these subsequences and properties. Using the union bound  allows us to show the desired properties in all subsequences in spite of the emerging dependencies. This is done in Lemma~\ref{seqwhp}.

We now show the concentration bounds on $Y_{i,k}$.
These bounds rely on the Chernoff-type bound  established in Lemma~\ref{lem:cherhoff-like}. This Chernoff-type bound  shows concentration for  variables having the property that the sum of the conditional probabilities of the variables, given all previous variables, is always  bounded (from above or below) by some $b$. The bound might be of general interest.% and its proof can be found in the appendix.
\begin{lemma}
    \label{lem:boundsY}
 Fix  configuration $\mathcal{C}_i$. Then,
   % Fix an arbitrary $i$ and consider the interval $I_{i,k}$.
        \begin{enumerate}[(a)]
      \item For
      $k= \gamma \frac{256d}{\alpha_1\cdot (1-\alpha_1)^2}$  with $\gamma \geq 1$  it holds that
        $$\Pr\left(Y_{i,k} < \gamma \right) \leq%\exp\left({- \frac{b}{144 (1-\alpha_1)}}\right)
% \exp\left(\frac{-\alpha_1 \cdot (1-\alpha_1)\cdot b}{64} \right)=
\exp\left(-\gamma \right).$$
        %where $0<\delta<1$, $c$ large enough, {\color{red} and $\alpha = \max\{\alpha_1,1/2\}$.}
%
\item   For $k\ge 0$, any
      $b'= \alpha_1\cdot(k+2d)/d$, and  any $ \delta > 0$ it holds that
      $$\Pr\left(Y_{i,k} < -(1+\delta)b'\right) \leq \exp
\Big(\tfrac{e^\delta}{(1+\delta)^{1+\delta}}\Big)^{b'}.$$
%\left(- \frac{\delta^2 b'}{3}\right).$$
    \end{enumerate}
\end{lemma}

\begin{proof}%[Proof of Lemma~\ref{lem:boundsY}]
%The core of the proof lies in establishing a Chernoff-type lower bound for $Y_{i,k}$.

First we show  a lower bound for $X_{i,k}$ and an upper bound for $X'_{i,k}$.
%The bounds for $X_{i,k}$ and $X'_{i,k}$ follow from the technical lemma stated below.

%\subsection{Proof of Lemma \ref{lem:boundsY}}
%The goal is to apply Lemma~\ref{lem:cherhoff-like}(b) to bound $X_{i,k}$.

\medskip

For $1\leq j\leq 2k+2d$,  let
$A_j $ be the event that $u_{i+j}$ switches its opinion from any opinion other than the preferred opinion to the preferred one.
Furthermore, for $1\leq j\leq 2k+2d$, let
$$
Z_j =  \begin{cases}
1 & \text{ if  $i+j\leq S_{i,k}$ and $A_j$ }\\
0 & \text{ otherwise }
\end{cases}
$$
\noindent
Define
\begin{align*}
P_j &=\begin{cases} \Pr(Z_j = 1 \mid Z_1,\ldots,Z_{j-1}) =\lambda_{i+j}/d  \phantom{0} & \text{  $i+j\leq S_{i,k}$ and $A_j$ }\\
0 &\text{ otherwise}
 \end{cases}
\end{align*}
\noindent
It follows that $X_{i,k} = \sum_j Z_j$.
We set $B=\sum_j P_i$. We derive $$B=\sum_j P_i = \sum_{s\in I_{i,k}}(1-o_s)\lambda_s/d = (\Lambda({S_{i,k}})-\Lambda(i))/d \geq k/d,$$
by definition of $I_{i,k}$.
From Lemma~\ref{lem:cherhoff-like}(b) with $b^*=k/d\leq B$ we obtain for any $0<\delta<1$, that 
\begin{align}\label{parta}
\Pr(X_{i,k} < (1-\delta)b^*) < e^{-b^*\cdot \delta^2/2} %= e^{-k\cdot \delta^2/(2d)}.
\end{align}

We now use a similar reasoning to bound $X'_{i,k}$.
For $1\leq j\leq 2k+2d$ let $A'_j $ be the event that $u_{i+j}$ switches its opinion from the preferred opinion to any other opinion.
Define for $1\leq j\leq 2k+2d$ the variables $Z'_j$ similar to $Z_j$ but using $A'_j$ instead of $A_j$. And similarly,
define for $1\leq j\leq 2k+2d$ the variables $P'_j$ similar as $P_j$ but using $Z'_j,\dots, Z'_1$ instead of $Z_j,\dots, Z_1$.
In the same spirit  as before one can define $B'=\sum_j P'_i$. We derive $$B'=\sum_j P'_i \leq \sum_{s\in I_{i,k}}o_s\cdot \alpha_1\cdot \lambda_s/d = \alpha_1\cdot (\Lambda'({S_{i,k}})-\Lambda'(i))/d \leq \alpha_1\cdot 	( k + 2d)/d,$$
where the last inequality follows from Observation~\ref{Ishort}.

Using a similar reasoning for $X'_{i,k}$ and applying Lemma~\ref{lem:cherhoff-like}(a) with $b' = \tfrac{\alpha_1\cdot (k+2d)}{d}\geq B'$, yields for $\delta>0$,
\begin{align}\label{partb}
\Pr\big(X'_{i,k} > (1+\delta)b'\big)
< \Big(\tfrac{e^\delta}{(1+\delta)^{1+\delta}}\Big)^{b'}
%= \Big(\tfrac{e^\delta}{(1+\delta)^{1+\delta}}\Big)^{\alpha_1\cdot(k+2d)/d}.
\end{align}
\noindent
To prove part~(a) we now combine \eqref{parta} and \eqref{partb}
 as follows.
Since  $k\geq 16d/(1-\alpha_1)$,  for $\delta=(1-\alpha_1)/8$,
we have that 
w.p. at least $1-e^{-b^*\cdot \delta^2/2} - \Big(\tfrac{e^\delta}{(1+\delta)^{1+\delta}}\Big)^{b'}$ that 

\begin{align}\label{asdf}
\begin{split}
Y_{i,k} &= X_{i,k} -X'_{i,k} \\
&\geq (1-\delta)b^* - (1+\delta)b' \\
&\geq \frac{k}{d}\left(1-\delta - (1+\delta)\alpha_1-(1+\delta)\frac{2d}{k} \right)  \\
&\geq \frac{k}{d}\left(1-\alpha_1 -2\delta -(1+\delta)\frac{2d}{k} \right)  \\
&\geq \frac{k}{d}\left(1-\alpha_1 -(1-\alpha_1)/4 -(1+\delta)\frac{2d}{k} \right)  \\
%\intertext{ }
&\geq \frac{k}{d}\left(1-\alpha_1 -(1-\alpha_1)/4-(1-\alpha_1)/4\right)\\
&\geq (1-\alpha_1)\frac{k}{2d} \\
&= \gamma \tfrac{32}{\alpha_1\cdot (1-\alpha_1)}\\
&> \gamma.
\end{split}
\end{align}
\noindent
We proceed by bounding $1-e^{-b^*\cdot \delta^2/2} - \Big(\tfrac{e^\delta}{(1+\delta)^{1+\delta}}\Big)^{b'}$.
Let $p=e^{-b^*\cdot \delta^2/2} +\Big(\tfrac{e^\delta}{(1+\delta)^{1+\delta}}\Big)^{b'}$. We have

\begin{align}\label{asdf2}
\begin{split}
 p&\leq  e^{-b^*\cdot \delta^2/2} + \Big(\tfrac{e^\delta}{(1+\delta)^{1+\delta}}\Big)^{b'} \\
&\leq e^{-b^*\cdot \delta^2/2} + e^{-b'\cdot \delta^2/2}\\
&= e^{-b^*\cdot (1-\alpha_1)^2/128} + e^{-b'\cdot  (1-\alpha_1)^2/128}\\
&\leq 2e^{-\alpha_1 \cdot (1-\alpha_1)^2\cdot k/(128d)}\\
&\leq e^{-\alpha_1 \cdot (1-\alpha_1)^2\cdot k/(266d) }\\
&= e^{-\gamma },
\end{split}
\end{align}
where we used that $k\geq \frac{256d}{\alpha_1\cdot (1-\alpha_1)^2}$.
Part~(a) follows from   \eqref{asdf} and \eqref{asdf2}.

\medskip
We now prove part~(b). Since we used that $Y_{i,k} \geq -X'_{i,k}$, the bound from \eqref{partb} implies  \begin{align*}
\Pr(Y_{i,k} < -(1+\delta)b') &\leq \Pr(X'_{i,k} > (1+\delta)b')\
< \Big(\tfrac{e^\delta}{(1+\delta)^{1+\delta}}\Big)^{b'}\\
%&\leq \exp\left(- \frac{\delta^2 b'}{3}\right),
\end{align*}
Hence, part~(b) of the lemma follows.
This completes the proof of Lemma~\ref{lem:boundsY}.

\end{proof}

The following two lemmas imply the theorem.
\begin{lemma}\label{seqwhp}
Let $\mathcal{S}$ be a step sequence. Then $\mathcal{S}$ is good with high probability.
\end{lemma}

\begin{proof}%[Proof of Lemma~\ref{seqwhp}]
We show for every property of Definition~\ref{def} that it holds with high probability. Fix an arbitrary $i\leq T'$.
\begin{enumerate}[(a)]
\item We derive from Lemma~\ref{lem:boundsY}(a) with  $k=2n \cdot\beta'$  that
        $\Pr(Y_{0,k} <  2n) \leq e^{-2n} %\exp\left({- \frac{b}{144 (1-\alpha_1)}}\right)
 \leq n^{-\beta}.$
\item Follows from $(c)$.
\item

Fix an arbitrary $k \leq T'$. We distinguish between two cases
\begin{enumerate}[(i)]
\item Case $k\leq \beta\cdot \log n\cdot \beta'$.
We derive from Lemma~\ref{lem:boundsY}(b) the following.
  For $k\ge 0$, any
      $b'= \alpha_1\cdot(k+2d)/d \leq \frac{300\beta\log n}{(1-\alpha_1)^2}=\ell/2$, and  any $\delta > 0$ it holds that
      $$ \Pr\left(Y_{i,k} < -(1+\delta)b'\right) \leq \exp
\Big(\tfrac{e^\delta}{(1+\delta)^{1+\delta}}\Big)^{b'}.$$

We distinguish once more between two cases.

\begin{itemize}
\item If $ \sqrt{\frac{3\beta \log n}{b'}} < 1$ set $\delta = \sqrt{\frac{3\beta \log n}{b'}}<1$. We have
\begin{align*}
 \Pr\left(Y_{i,k} <  -\ell \right) &\leq \Pr\left(Y_{i,k} < -2b'\right) \leq \Pr\left(Y_{i,k} < -(1+\delta)b'\right)\\
&\leq \exp
\Big(\tfrac{e^\delta}{(1+\delta)^{1+\delta}}\Big)^{b'}\leq \exp\left(-\delta^2 b'/3\right)\leq
%\exp\left( b'/4\right)\leq
n^{-\beta}.
\end{align*}
\item Otherwise, we have $b' \leq 3\beta\log n < 4\beta \log n$. Set $\delta = \frac{4\beta \log n}{b'}>1$.
We have
$$ Pr\left(Y_{i,k} <-\ell \right) \leq \Pr\left(Y_{i,k} < -(1+\delta)b'\right) \leq \exp
\Big(\tfrac{e^\delta}{(1+\delta)^{1+\delta}}\Big)^{b'}\leq \exp\left(-\delta b'/3\right)\leq
%\exp\left( b'/4\right)\leq
n^{-\beta}.$$
\end{itemize}

Thus, the claim follows.

\item Case $k> \beta\cdot \log n\cdot \beta'$.
We derive from Lemma~\ref{lem:boundsY}(a)  that
   $\Pr(Y_{i,k} <  -\ell ) \leq\Pr(Y_{i,k} <  \beta\log n) < %\exp\left(-\frac{\alpha_1 \cdot (1-\alpha_1)^2\cdot k}{64d} \right)
 n^{-\beta}.$

\end{enumerate}
Hence, in all cases we have $\Pr(Y_{i,k} <  -\ell)\leq n^{-\beta}$.
\item
We derive from Lemma~\ref{lem:boundsY}(a) with  $k= z \cdot \beta'$  that
 $\Pr(Y_{i,k} < z ) < \exp(-z)\leq n^{-\beta}.$

\end{enumerate}
Since $0<\alpha_1<1$, we have that $T'\leq n^3$.
The number of events in Definition~\ref{def} is bounded by
$5 T'^2 \leq n^7$.
Thus choosing, $\beta \geq 8$, and taking union bound over all these events yields the claim.
\end{proof}

\begin{lemma}\label{seqwin}
%Fix $\mathcal{C}_{T'}$.
If $\mathcal{S}$ is a good  step sequence, then in at most $T'$ time steps, the preferred opinion prevails and the $T'$ time steps occur before round  $\tau'''$.
\end{lemma}

\begin{proof}%[Proof of Lemma~\ref{seqwin}]

Recall, that  we assume that if in some round  the preferred opinions vanishes, then after this round, the opinion of some fixed node switches spontaneously.
Similarly, if in some round  the preferred opinion prevails, then after this round, the opinion of some fixed node switches spontaneously to an arbitrary other opinion.the preferred opinion never vanishes.
This process $P'$ diverges from the original process $P$ only after the first step where either the preferred opinion prevails or vanishes.
From $(a)$ and $(b)$ of the definition of a good sequence
it follows that the preferred opinion prevails in $P'$ after $T'$ steps.
It is easy to couple both process so that the good opinion also prevails in the original process $P$.

 \medskip
It remains to argue that the $T'$ time steps happen before round $\tau'''$. Using the definition of the conductance, we can lower bound the number of steps in any round $t$ by $|cut(S_t,S'_t)| \geq d\cdot \min\{|S_t|, |S'_i|\}\cdot \phi_t$.
We then consider intervals of sufficient length in which the size of the preferred doubles as long as its size is below $n/2$. Afterward, one can argue that size of all the non preferred opinions halves every interval.
We now give some intuition for the remainder of the proof.
Consider the following toy case example of a static graphs  with $\alpha_1 =0$ (rumor spreading). The length of an interval required for the preferred $S$ with $|s| \leq n/3$ to double in expectation  is bounded by
$1/\phi$. In our setting, we need to handle  two main difference w.r.t. the toy case example.
First, the number of nodes with the preferred opinion can  reduce by up to $\beta \log n$ (Definition~\ref{def}(c)). Since $\beta$ is constant, this can be easy compensated by slightly longer intervals.
% in which we can even guarantee that the $S$ doubles (This argument builds on Definition~\ref{def}(e)).
Second, the graph is dynamic as opposed to static. To address this we 'discretize', similarly as before, rounds into consecutive phases which
ensure that sum of the $\phi_t$ for rounds $t$ in this phase is at least $1$. Thus, in our toy example one requires  $1 $ phases in expectation.
\medskip

We proceed by discretizing the  rounds into phases.
 Phase $i$ starts at round  $\tau(i)=\min \{ t: \sum_{j=1}^t \phi_j \geq 2i \} $ for $i\geq 0$ and it ends at round $\tau({i+1})-1$. Since $\phi_j \leq 1$ for all $j\geq 0$ we have $\tau(0) < \tau(1) < \dots$ and $\sum_{j=\tau_i}^{\tau({i+1})-1} \phi(j) \geq 1$ for $i\geq 0$.
We now map the steps of $\mathcal{S}$ to rounds. For this we define the \emph{check point} $t_j$ to be the following round for $j\leq j_{max}= 4\logn+1$.

$$t_j=
\begin{cases}
0 & \text{ if $j=0$  }\\
\tau\left( 12\ell \cdot \beta'/d \right) & \text{ if $j=1$  }\\
\tau\left( 12\ell \cdot\beta'/d  + (j-1)\cdot 24\beta'/d \right)  & \text{ if $2\leq j \leq 4 \log n$  }\\
\tau\left(  24\ell \cdot \beta'/d  + (j_{max}-1) \cdot 24\beta'/d\right)  & \text{ if $j=j_{max}$  }\\
\end{cases}
$$

Given any good sequence $\mathcal{S}$, we show by induction over $j$  that the following lower bounds on the size of the preferred opinion at these check points. More concretely, define for all $j\leq j_{max}$  that

\begin{align*}
%\begin{split}
\zeta(j) =
\begin{cases}
0  & \text{ if $j=0$ }\\%
2\ell  & \text{ if $j=1$ }\\%\phantom{00000000000000000000}(case 1)
\min\{2\ell \cdot  2^{j-2}, n/2 \} & \text{ if $2\leq j \leq 2\log  n$  }\\ %\phantom{00000000}(case 2)
\min\{ n - 2^{\log (n/2)-(j-2\log n)}, n-2\ell\} & \text{ if $2\log n  < j \leq 4\log n$ }\\ %(case 3)
n & \text{ if $j = j_{max}$} %\phantom{0000000}(case 4) }
\end{cases}
%\end{split}
\end{align*}
We now show for all $j \leq j_{max}$ that 
\begin{align}\label{checkpoints}
\text{ $|s_{t_j}| \geq \zeta(j),$  }
\end{align}
We consider each of the cases depending of the size of $j$ w.r.t.  \eqref{checkpoints}.
The induction hypothesis $j=0$ holds trivially. Suppose the claim holds for $j-1$ for $j\geq 1$.

We assume w.l.o.g. the following about $\mathcal{S}$:
%First, assume there is no round $t \leq t_j$
%with a time step $i$ such that $Y_{0,i} \geq 2n$ since this  already implies, as argued before, that consensus is reached at round $t\leq t_{j_{max}}\leq \tau'''$, would yield the claim.
%Second, assume 
there is no step $t$ before round $ t_j$ with $|s_{t}|\geq \zeta(j)+\ell$ since by Definition~\ref{def}(c) this implies that $|s_{t_j}| \geq \zeta(j)$ which yields the inductive step.
This  assumption, implies that for all $t\leq t_j$ we have $|s_{t}| < |\zeta(j)|+\ell$.
On the other hand, by induction hypothesis and Definition~\ref{def}(c) we have $ \zeta({j-1})-\ell\leq |s_{t}|$.
Thus, we assume in the following
\begin{align}\label{tamedsize}
|s_t| \in [\zeta({j-1})-\ell, \zeta(j)+\ell]
\end{align}

\noindent
We now distinguish between the following cases based on $j$.

\begin{itemize}
\item $j=1:$ In every step $t\in (0, t_1]$ the number of edges  crossing the cut is at least  $\phi_t\cdot  d$ and hence the  number of edges  crossing the cut in the interval  $(0, t_1]$ is at least $\sum_{i=1}^{t_1} \phi_i\cdot d \geq 12\ell \cdot \beta' $. Let $k=3\ell \cdot \beta' $. Definition~\ref{def}(d) implies that $ Y_{1,k} >3\ell$. Hence, by Definition~\ref{def}(c) we have $|s_{t_1}|\geq \zeta(1)$ as desired.

\item $2\leq j \leq 2\log n:$ In every step $t\in (t_{j-1}, t_j]$
we have, by \eqref{tamedsize}, that the number of edges crossing the cut is at least $$  \phi \cdot d \cdot (\min\{|s_t|,|s'_t|\}-\ell) \geq \phi \cdot d \cdot \min\{\zeta(j-1)-\ell , n/2 - \ell\} \geq \phi \cdot d\cdot (\zeta(j-1)-\ell)\geq \phi_t \cdot d\cdot \zeta(j-1)/2.$$ Hence the number of edges  crossing the cut in the interval  $(t_{j-1}, t_j]$ is at least $ 12\zeta(j-1)\cdot \beta' $. Let $k=3\zeta(j-1)\cdot \beta' $. Definition~\ref{def}(d) implies that $ Y_{t_{j-1},k} \geq 3\zeta(j-1)\geq \zeta(j)+\ell$. Hence, by Definition~\ref{def}(c) we have $|s_{t_{j}}| \geq \zeta({j})$.

\item $2\log n  < j \leq 2\log n:$ In every step $t\in (t_{j-1}, t_j]$
we have, by \eqref{tamedsize}, that the number of edges  crossing the cut is at least $$\phi_t \cdot d \cdot (\min\{|s_t|,|s'_t|\}-\ell) \geq \phi \cdot d\cdot (n-\zeta(j)-\ell)\geq\phi \cdot  d\cdot (n-\zeta(j))/2.$$ Hence the number of edges  crossing the cut in the interval  $(t_{j-1}, t_j]$ is at least $\sum_{i=t_{j-1}+1}^{t_{j}} \phi_i\cdot d\cdot (n-\zeta(j))/2 \geq 12(n-\zeta(j))\cdot \beta' $. Let $k=3(n-\zeta(j))\cdot \beta' $. Definition~\ref{def}(d) implies that $ Y_{t_{j-1},k} \geq 3(n-\zeta(j))\geq (n-\zeta(j))+\ell$. Hence, by Definition~\ref{def}(c) we have $|s_{t_{j}}| \geq \min\{\zeta(j-1)+(n-\zeta(j)),n-2\ell\} \geq \zeta(j)$.

\item $j_{max}:$ In every step $t\in (t_{j-1}, t_j]$
we have, by \eqref{tamedsize}, that the number of edges crossing the cut at any time step is at least $\phi_t  \cdot d.$ Hence the number of edges  crossing the cut in the interval  $(t_{j-1}, t_j]$ is at least $ d \geq 12\ell\cdot \beta' $. Let $k=3\ell\cdot \beta' $. Definition~\ref{def}(d) implies that $ Y_{t_{j_{max}-1},k} \geq 3\ell $. Hence, by Definition~\ref{def}(c) we have $|s_{t_{j_{max}}}| \geq n = \zeta({j_{max}})$.

\end{itemize}
\noindent
This completes the proof of \eqref{checkpoints}. We have $$t_{j_{max}}= \tau\left( 24\ell \cdot \beta' + j_{max} \cdot 24\beta'/d \right) \leq 4\left(24\ell \cdot \beta' + (j_{max}-1) \cdot 24\beta'/d\right)\leq \tau''', $$
which yields the proof.

\end{proof}

\begin{proof}[Proof of Theorem~\ref{thm:ub-bvoter}]
The claim follows from Lemma~\ref{seqwhp} together with Lemma~\ref{seqwin}.
\end{proof}
\newpage

\bibliographystyle{abbrv}
\bibliography{biblio}

\begin{thebibliography}{10}

\bibitem{AD15}
M.~A. Abdullah and M.~Draief.
\newblock Global majority consensus by local majority polling on graphs of a
  given degree sequence.
\newblock {\em Discrete Applied Mathematics}, 180:1--10, 2015.

\bibitem{AF14}
D.~Aldous and J.~Fill.
\newblock Reversible markov chains and random walks on graphs, 2002.
\newblock Unpublished,
  \url{http://www.stat.berkeley.edu/~aldous/RWG/book.html}.

\bibitem{ASS08}
N.~Alon, O.~Schwartz, and A.~Shapira.
\newblock An elementary construction of constant-degree expanders.
\newblock {\em Comb. Probab. Comput.}, 17(3):319--327, May 2008.

\bibitem{AKL08}
C.~Avin, M.~Kouck{\'{y}}, and Z.~Lotker.
\newblock How to explore a fast-changing world ({C}over time of a simple random
  walk on evolving graphs).
\newblock In {\em Proceedings of the 35th International Colloquium on Automata,
  Languages and Programming (ICALP)}, pages 121--132, 2008.

\bibitem{ABKU99}
Y.~Azar, A.~Broder, A.~Karlin, and E.~Upfal.
\newblock Balanced allocations.
\newblock {\em SIAM J. Comput.}, 29(1):180--200, 1999.

\bibitem{CEOR12}
C.~Cooper, R.~Els\"{a}sser, H.~Ono, and T.~Radzik.
\newblock Coalescing random walks and voting on connected graphs.
\newblock {\em {SIAM} J. Discrete Math.}, 27(4):1748--1758, 2013.

\bibitem{CER14}
C.~Cooper, R.~Els{\"a}sser, and T.~Radzik.
\newblock The power of two choices in distributed voting.
\newblock In {\em Proceedings of the 41st International Colloquium on Automata,
  Languages and Programming (ICALP)}, pages 435--446, 2014.

\bibitem{CERRS15}
C.~Cooper, R.~Elsässer, T.~Radzik, N.~Rivera, and T.~Shiraga.
\newblock Fast consensus for voting on general expander graphs.
\newblock In {\em Proceedings of the 29th International Symposium on
  Distributed Computing (DISC)}, pages 248--262, 2015.

\bibitem{CG14}
J.~Cruise and A.~Ganesh.
\newblock Probabilistic consensus via polling and majority rules.
\newblock {\em Queueing Systems}, 78(2):99--120, 2014.

\bibitem{DGMRSS14}
J.~D{\'{\i}}az, L.~Goldberg, G.~Mertzios, D.~Richerby, M.~Serna, and
  P.~Spirakis.
\newblock Approximating fixation probabilities in the generalized moran
  process.
\newblock {\em Algorithmica}, 69(1):78--91, 2014.

\bibitem{DW83}
P.~Donnelly and D.~Welsh.
\newblock Finite particle systems and infection models.
\newblock {\em Mathematical Proceedings of the Cambridge Philosophical
  Society}, 94:167--182, 1983.

\bibitem{EFKMT16}
R.~{Els{\"a}sser}, T.~{Friedetzky}, D.~{Kaaser}, F.~{Mallmann-Trenn}, and
  H.~{Trinker}.
\newblock Efficient $k$-party voting with two choices.
\newblock {\em ArXiv e-prints}, Feb. 2016.

\bibitem{G11}
G.~Giakkoupis.
\newblock Tight bounds for rumor spreading in graphs of a given conductance.
\newblock In {\em Proceedings of the 28th International Symposium on
  Theoretical Aspects of Computer Science (STACS)}, pages 57--68, 2011.

\bibitem{GSS14}
G.~Giakkoupis, T.~Sauerwald, and A.~Stauffer.
\newblock Randomized rumor spreading in dynamic graphs.
\newblock In {\em Proceedings of the 41st International Colloquium on Automata,
  Languages and Programming (ICALP)}, pages 495--507, 2014.

\bibitem{HP01}
Y.~Hassin and D.~Peleg.
\newblock Distributed probabilistic polling and applications to proportionate
  agreement.
\newblock {\em Information and Computation}, 171(2):248--268, 2001.

\bibitem{HL75}
R.~Holley and T.~Liggett.
\newblock Ergodic theorems for weakly interacting infinite systems and the
  voter model.
\newblock {\em The Annals of Probability}, 3(4):643--663, 1975.

\bibitem{HV11}
B.~Houchmandzadeh and M.~Vallade.
\newblock The fixation probability of a beneficial mutation in a geographically
  structured population.
\newblock {\em New Journal of Physics}, 13(7):073020, 2011.

\bibitem{KT08}
M.~Kearns and J.~Tan.
\newblock Biased voting and the democratic primary problem.
\newblock In {\em Proceedings of the 4th International Workshop on Internet and
  Network Economics (WINE)}, pages 639--652, 2008.

\bibitem{KLO10}
F.~Kuhn, N.~Lynch, and R.~Oshman.
\newblock Distributed computation in dynamic networks.
\newblock In {\em Proceedings of the 42nd {ACM} Symposium on Theory of
  Computing (STOC)}, pages 513--522, 2010.

\bibitem{LLMSW12}
H.~Lam, Z.~Liu, M.~Mitzenmacher, X.~Sun, and Y.~Wang.
\newblock Information dissemination via random walks in $d$-dimensional space.
\newblock In {\em Proceedings of the 23rd Annual ACM-SIAM Symposium on Discrete
  Algorithms (SODA)}, pages 1612--1622, 2012.

\bibitem{LN07}
N.~Lanchier and C.~Neuhauser.
\newblock Voter model and biased voter model in heterogeneous environments.
\newblock {\em Journal of Applied Probability}, 44(3):770--787, 2007.

\bibitem{L85}
T.~Liggett.
\newblock {\em Interacting Particle Systems}.
\newblock Springer Berlin Heidelberg, 1985.

\bibitem{MS13}
G.~Mertzios and P.~Spirakis.
\newblock Strong bounds for evolution in networks.
\newblock In {\em Proceedings of the 40th International Colloquium on Automata,
  Languages and Programming (ICALP)}, pages 669--680, 2013.

\bibitem{MU05}
M.~Mitzenmacher and E.~Upfal.
\newblock {\em Probability and Computing: Randomized Algorithms and
  Probabilistic Analysis}.
\newblock Cambridge Univ.\ Press, 2005.

\bibitem{MP95}
R.~Motwani and P.~Raghavan.
\newblock {\em Randomized Algorithms}.
\newblock Cambridge University Press, 1995.

\bibitem{NIY00}
T.~Nakata, H.~Imahayashi, and M.~Yamashita.
\newblock A probabilistic local majority polling game on weighted directed
  graphs with an application to the distributed agreement problem.
\newblock {\em Networks}, 35(4):266--273, 2000.

\bibitem{O12}
R.~Oliveira.
\newblock On the coalescence time of reversible random walks.
\newblock {\em Transactions of the American Mathematical Society},
  364(4):2109--2128, 2012.

\bibitem{PSSS13}
Y.~Peres, A.~Sinclair, P.~Sousi, and A.~Stauffer.
\newblock Mobile geometric graphs: {D}etection, coverage and percolation.
\newblock {\em Probability Theory and Related Fields}, 156(1-2):273--305, 2013.

\end{thebibliography}

\clearpage

\appendix

\section{Auxiliary Claims}
\begin{lemma}\label{lem:replaceRV}
%Let $X_1 \kdots X_k$ and $Y_1 \kdots Y_k$ be a sets of independent variables.
Let $X_i$ and $Y_i$ be the random variables defined in the proof of Lemma~\ref{lem:technical-st}.
Let $f(\cdot)$ be a concave and continuous function.
We have \\$E\left[f\left(\sum_i X_i\right)\right]\leq E\left[f\left(\sum_i Y_i\right)\right]. $
\end{lemma}
\begin{proof}%[\bf Proof of Lemma \ref{lem:replaceRV}]
We show by induction over $i$ that the variables $$E[f(Y_1 + \cdots + Y_{i-1} + X_{i} + \cdots + X_k)] \leq E[f(Y_1 + \cdots + Y_{i} + X_{i+1} + \cdots + X_k)].$$
Let $Z = Y_1 + \cdots + Y_{i-1} + X_{i+1} + \cdots + X_k$.
In step $i\rightarrow i+1$ we have
\begin{align*}
E[&f(Y_1 + \cdots + Y_{i-1} + X_{i} + \cdots + X_k)] \\
&=E[f(Z + X_{i})] \\
&= \frac{\lambda_i}{\de_{i}} E[f(Z+ d_{i} ]  + \left(1-\frac{\lambda_i}{ \de_{i}}\right)E[f(Z)] \\
&= \lambda_i E\left[\frac{f(Z + d_{i})}{\de_{i}}\right]  + \left(1-\frac{\lambda_i}{ \de_{i}}\right)E\left[f(Z)\right]  \\
&=\lambda_i E\left[\frac{f(Z+ d_{i} )  - f(Z)
}{\de_{i}}\right]  + E\left[f(Z)\right]  \\
&\leq  \lambda_i E\left[\frac{f(Z+ \lambda_{i} )  - f(Z)
}{\lambda_{i}}\right] + E\left[f(Z)\right] \\
&= E[f(Y_1 + \cdots + Y_{i}+X_{i+1} + \cdots + X_k)],
\end{align*}
where the last inequality follows from the concavity of $f(\cdot)$.
\end{proof}
%The proof can be found in the full version \cite{master}.

\begin{lemma}
    \label{lem:skewness}
    Let $Z_1,\ldots,Z_n$ be independent random variables and $Z = \sum_i Z_i$.
    If $\Exp[Z]=0$, then $\Exp[Z^3] = \sum_i\big(\Exp[Z_i^3]-3\Exp[Z_i^2]\cdot\Exp[Z_i] + 2\Exp[Z_i]^3\big)$.
\end{lemma}
\begin{proof}
%The proof can be found in the appendix.%
%\subsection{Proof of Lemma~\ref{lem:skewness}}
In the following we make use of $\Exp[\sum_i Z_i] = 0$. We derive
\begingroup
\allowdisplaybreaks
\begin{align*}
 \Exp\left[ Z ^3\right] &=  \Exp\left[ \sum_{i,j,k}Z_i Z_j Z_k \right]
\\&=
\sum_{i}\Exp\left[ Z_i\sum_{j,k} Z_j Z_k \right]
\\&=
\sum_{i}\Exp\left[Z^3_i\right] %\\&\mspace{15mu}
+3\sum_{i}\Exp\left[Z^2_i \right]\Exp\left[\sum_{j,\ j\not= i}Z_j \right] \\&\mspace{15mu}
+ \sum_{i}\Exp\left[Z_i \right] \sum_{j,\ j\not= i}\Exp\left[Z_j \right]   \Exp\left[   \sum_{k,\ k\not= i,j}     Z_k\right]  \\&= %%%%%%%%%%%%%%%
\sum_{i}\Exp\left[Z^3_i\right] %\\&\mspace{15mu}
+  3\sum_{i}\Exp\left[Z^2_i\right]( 0-\Exp[Z_i]   )  \\&\mspace{15mu}
+
 \sum_{i}\Exp\left[Z_i\right] \sum_{j,\ j\not= i}\Exp\left[Z_j\right]  (0-\Exp[Z_i]-\Exp[Z_j] )\\&= %%%%%%%%%%%%%%%
\sum_{i}\Exp\left[Z^3_i\right] %\\&\mspace{15mu}
- 3\sum_{i}\Exp\left[Z^2_i \right]\Exp\left[Z_i \right] % \\&\mspace{15mu}
+ \sum_{i}\Exp\left[Z_i\right](-\Exp[Z_i]) \sum_{j\ j\not= i}\Exp\left[Z_j\right] \\&\mspace{15mu}+
 \sum_{i}\Exp\left[Z_i\right] \sum_{j\ j\not= i}\Exp\left[Z_j\right](-\Exp[Z_j]) \\&=%%%%%%%%%%%%%%%
\sum_{i}\Exp\left[Z^3_i\right] %\\&\mspace{15mu}
- 3\sum_{i}\Exp\left[Z^2_i \right]\Exp\left[Z_i \right] %\\&\mspace{15mu}
+
 2\sum_{i}\Exp\left[Z_i\right](-\Exp[Z_i]) (-\Exp[Z_i])  %\\&+
% \sum_{i}\Exp\left[Z_i\sum_{v\ v\not= i} \Exp\left[Z_j\right](-\Exp[Z_j]) \right]
 \\&=%%%%%%%%%%%%%%%
\sum_{i}\Exp\left[Z^3_i\right] -%\\&-&
 3\sum_{i}\Exp\left[Z^2_i \right]\Exp\left[Z_i \right] +% \\&+&
2 \sum_{i}\Exp[Z _i]^3.%%%%%%%%%%%%%%%%%%%%%%
\end{align*}%
\endgroup
\end{proof}

We now show $d$-regular graphs with a cut of size $\Theta(\phi_t d n)$ exist indeed.
\begin{lemma}\label{lem:givemecut}
Let $\tfrac1{nd}\leq\phi\leq 1$. Let $0 <\gamma < 1$ be some constant.
Let $d\geq 6$ be an even integer.
For any integer $n'\in [\gamma n, n/2]$ there exists a $d$-regular graph $G=(V,E)$ with $n$ nodes and the following property.
There is a set $S\subset V$ with $|S|=n'$ such that  $|cut(S, V\setminus S)|
=\Theta(\phi d n )$.
%\in [c \phi d n, \phi d n/4 ]$ for some constant $c$.
Moreover, there are at least $n'/2$ nodes without any edges in $cut(S, V\setminus S)$.
\end{lemma}
\begin{proof}
%Let $G'$($G''$ respectively) be any $d$-regular directed symmetric graph of size $n'$($n-n'$, respectively).
In the following we create two $d$-regular graphs $G'$ and $G''$ and connect them to a $d$-regular graph $G$ such that the cut size is $\Theta(\phi d n)$.
Let $G'=(V',E')$ be the circulant graph $C_{n'}^{\floor*{d/2}}$ with $V'=\{ v'_1, \dots, v'_{n'} \}$.
Let $G''=(V'',E'')$ be the circulant graph $C_{n-n'}^{\floor*{d/2}}$ with $V'=\{ v''_1, \dots, v''_{n-n'} \}$.
We now connect $G'$ and $G''$. W.l.o.g. $\phi  \leq 1/d$. The case $\phi  > 1/d$ is analogue.
We choose $k$ such that we have $2\leq k < n'/2$ and $k=\Theta(\phi d n')$.
Let $S'=\{v'_1, \kdots, v'_k \}$ and let $S''=\{ v''_1, \dots, v''_k \}$.
%Let $S'$ be a  subset of nodes of $V'$ with $|S|=k$ such that the induced subgraph contains a simple path $s'_1,\dots, s'_k$, where $V'$ are the vertices of $G'$.
Now we remove all edges $(v'_{i},v'_{i+1})$ with $1\leq i\leq k-1$ and $(v''_{i},v''_{i+1})$. Note that $G'$($G''$ respectively) is still connected.
Furthermore, the vertices $v'_1$, $v'_k$, $v''_1$, and $v''_k$ have degree $d-1$ and all other vertices of $S'\cup S''$ have degree $d-2$.
One can easily add $(2k-2)$ edges  such that
$(i)$ one endpoint of every edge is in $V'$ and one in $V''$, and $(ii)$ all vertices in $G'$ and $G''$ have degree $d$.

Note that the obtained graph $G$ is connected and the cut $|cut(V', V\setminus V')|$ contains $(2k-2)=\Theta(\phi d n' )$ edges. The claim follows directly.

\end{proof}

\begin{theorem}[{\cite[Theorem 12.2]{MU05}}]\label{optional}
If $Z_0, Z_1, \dots$ is a martingale with respect to $X_1, X_2, \dots$ and if $T$ is a stopping time for $X_1, X_2, \dots$, then
$$E[Z_t]=E[Z_0] $$
whenever one of the following holds:
\begin{itemize}
\item The $Z_i$ are bounded, so there is a constant for all $i$, $|Z_i| \leq c$;
\item $T$ is bounded;
\item $E[T] \leq \infty$, and there is a constant $c$ such that $E[Z_{t+1} - Z_t | X_1, \dots X_t] < c$;
\end{itemize}
\end{theorem}

\subsection{A Chernoff-type bound}
The following lemma bounds the sum of (dependent) binary random variables, under the assumption that the sum of the conditional probabilities of the variables, given all previous variables, is always  bounded (from above or below) by some $b$.
The bounds are the same as the ones for independent variables but use $b$ in place of $\mu$. The bound can be seen as a generalisation of \cite{ABKU99}.
The proof follows the proof of the independent case.

\begin{lemma}[Chernoff Bound for Dependent Setting]
    \label{lem:cherhoff-like}
    Let $Z_1,Z_2,\ldots,Z_\ell$ be a sequences of binary random variables, and for each $1\leq i\leq  \ell$, let $p_i = \Pr(Z_i = 1 \mid Z_1,\ldots,Z_{i-1})$.
    Let $Z=\sum_i Z_i$, $B=\sum_i p_i$.
    \begin{enumerate}[(a)]
      \item For any $b\ge 0$ with
       $\Pr(B \leq b) = 1$, it holds for any  $\delta > 0$ that
      $$\Pr(Z > (1+\delta)\cdot b) < \left(\frac{e^\delta}{(1+\delta)^{1+\delta}}\right)^b.$$
      \item For any $b\ge 0$ with
       $\Pr(B \geq b) = 1$, then for any $0<\delta<1$ it holds that
       $$\Pr(Z < (1-\delta)b) < e^{-b\delta^2/2}.$$
    \end{enumerate}
\end{lemma}
\begin{proof}%[Proof of Lemma~\ref{lem:cherhoff-like}]
The proof follows the proof of the Chernoff bound given in ~\cite{MP95}. However, the random variables $\{Z_i : 1\leq i\leq k\}$ are not independent.
For any positive real $t$ we have
$\Pr(Z \geq (1+\delta)b)=\Pr(exp(t Z) \geq exp(t (1+\delta)b) ) $. Thus, by applying Markov inequality we derive
\begin{equation}\label{aftermarkov}
\Pr(Z \geq (1+\delta)b) < \tfrac{\Exp[exp(t Z)]}{exp(t (1+\delta)b)}.
\end{equation}
We now bound  $\Exp[exp(t Z)] $. By law of total expectation and $p_1=\Pr(Z_1 = 1)$ we get

\begin{align*}
\Exp&[exp(t Z)] =\Exp[ exp(t Z) | Z_1 = 1 ]  \Pr(Z_1 = 1) + \Exp[ exp(t Z) | Z_1 = 0 ]  \Pr(Z_1 = 0)\\
%%%%%%
&=\Exp[ exp(t Z) | Z_1 = 1 ]  p_1 + \Exp[ exp(t Z) | Z_1 = 0 ]  (1-p_1)\\
%%%%%%%%
&=\Exp\left[ exp\left(t \sum\limits_{i=2}^kZ_i\right) \middle| Z_1 = 1 \right]  exp(t)  p_1 + \Exp\left[ exp\left(t \sum\limits_{i=2}^kZ_i\right) \middle| Z_1 = 0 \right]  (1-p_1)\\
%%%%%%%%
&\leq\max\left\{\Exp\left[ exp(t \sum\limits_{i=2}^kZ_i) \middle| Z_1 = 1 \right], \Exp\left[ exp\left(t \sum\limits_{i=2}^kZ_i\right) \middle| Z_1 = 0 \right]\right\}  (p_1   exp(t) + 1-p_1).
\end{align*}
Repeating this inductively for the variables $Z_2, \dots, Z_k$ yields
\begin{align*}
\Exp[exp(t Z)] &< \prod_i (P_i   exp(t) + 1-P_i)=
 \prod_i (1+P_i(exp(t)-1)).
\end{align*}
Using $1+ x < e^x$ and rearranging gives
$\Exp[exp(t Z)] < exp((exp(t)- 1)b). $
Plugging this into Eq.~\eqref{aftermarkov} yields Claim (a).
Claim~(b) can be proven analogously.

\end{proof}

%\section{Omitted Proofs of Section~\ref{standardvotermodel}}
%Astan

%lastLAST

%
%\section{Omitted Proofs of Section~\ref{biasedvotermodel}}\label{aux}%Abias

%Ebias (personal marker Frederik)

%\section*{APPENDIX}
%\renewcommand{\thesubsection}{A}
%
%\section{Omitted Proofs}%\label{OMstandardvotermodel}
%

%\section{Remove later - I already saved it elsewhere so feel free to delete it}
%\begin{table}[H]
%\begin{tabular}{llllll}
%\hline
%Class                     & Mixing Time   & Meeting Time  & Hitting Time            & Our bound            & \cite\{CEOR12\} \\ \hline
%Cycle                     & $\Theta(n^2)$ & $\Theta(n^2)$ & $\Theta(n^2)$           & $O(n^2)$             & $O(n^3)$        \\ \hline
%$C_n^k$, for constant $k$ & $\Theta(n^2)$ & $\Theta(n^2)$ & $\Theta(n^2)$           & $O(n^2)$             & $O(n^3)$        \\ \hline
%Grid                      & $\Theta(n)$   & ?             & $\Theta(n\cdot \log n)$ & $O(n\cdot \sqrt{n})$ & $O(n^2)$        \\ \hline
%Binary Tree               & $\Omega(n)$   & ?             & ?                       & $O(n^2)$             & $n^2$           \\ \hline
%Power Law Graph           & ?             & ?             & ?                       & ?                    & ?               \\ \hline
%\end{tabular}
%\end{table}

%$$ E[Z_{t+1}]= \sum_{\phi} \sum_{z} E[Z_{t+1} | Z_t=z,\Phi_t=\phi] Pr(\Phi_t=\phi|Z_t=z)\cdot Pr(Z_t=z)    $$
%$$ E[Z_{t+1}-z|Z_t=z]= \sum_{\phi} E[Z_{t+1} -z | Z_t=z,\Phi_t=\phi] Pr(\Phi_t=\phi)=\sum_\phi \frac{\phi}{z} Pr(\Phi_t=\phi)= E[\Phi_t]/z $$

\end{document}